\documentclass[a4paper, 12 pt, twoside, reqno]{amsart}

\usepackage[english]{babel}
 
\usepackage[utf8]{inputenc}

\usepackage{setspace}
\usepackage[euler-digits,euler-hat-accent]{eulervm}

\usepackage{times}
\usepackage{amsmath}
\usepackage{amsfonts}
\usepackage{amssymb}
\usepackage{amsmath}
\usepackage{amsthm}
\usepackage{graphicx}
\usepackage{array}
\usepackage{color}
\usepackage{mathrsfs}
\usepackage{hyperref}
\usepackage{eucal}
 
\usepackage{esint}  
\usepackage{tikz}
\usepackage{upgreek}
\usepackage{enumitem}

\allowdisplaybreaks

\theoremstyle{plain}
\newtheorem{theorem}{Theorem}[section]

\newtheorem{proposition}[theorem]{Proposition}
\newtheorem{lemma}[theorem]{Lemma}

\theoremstyle{definition}
\newtheorem{definition}[theorem]{Definition}
\newtheorem{assumption}[theorem]{Assumption}

\theoremstyle{remark}
\newtheorem{remark}[theorem]{Remark} 
\newtheorem{example}[theorem]{Example}

\numberwithin{equation}{section}
\numberwithin{figure}{section}
\numberwithin{table}{section}

\newcommand{\R}{\mathbb{R}}
\newcommand{\N}{\mathbb{N}}
\newcommand{\C}{\mathbb{C}}                           

\newcommand{\Z}{\mathbb{Z}}

\newcommand{\s}[1]{\CMcal{#1}}
                  
\newcommand{\bb}[1]{\mathscr{#1}}
\newcommand{\rr}[1]{\mathfrak{#1}}
\newcommand{\n}[1]{\mathbb{#1}}

\newcommand{\expo}[1]{\,\mathrm{e}^{#1}\,}                 

\newcommand{\dd}{\,\mathrm{d}}
\newcommand{ \ii}{\,\mathrm{i}\,}

\newcommand{\virg}[1]{\lq\lq#1\rq\rq}                \newcommand{\ie}{\textsl{i.\,e.\,}}
\newcommand{\eg}{\textsl{e.\,g.\,}}
\newcommand{\cf}{\textsl{cf}.\,}

\hypersetup{
pdftoolbar=true,        % show Acrobat’s toolbar?
pdfmenubar=true,        % show Acrobat’s menu?
pdffitwindow=true,     % window fit to page when opened
pdfstartview=true,    % fits the width of the page to the window
pdftitle={NCG-Landau-II},    % title
pdfauthor={Giuseppe De Nittis},     % author
breaklinks=true, %permet le retour a la ligne dans les liens trop longs
colorlinks=true,       % false: boxed links; true: colored links
linkcolor=purple,         % color of internal links (change box color 
citecolor=teal, % color of links to bibliography
urlcolor=blue, %couleur des hyperliens
bookmarksopen=true, %si les signets Acrobat sont crees,
filecolor=magenta,      % color of file links
}

\begin{document}

\title[The cohomology invariant for class DIII topological insulators]{The cohomology invariant for class DIII topological insulators}

\author[G. De~Nittis]{Giuseppe De Nittis}

\address[G. De~Nittis]{Facultad de Matem\'aticas \& Instituto de F\'{\i}sica,
  Pontificia Universidad Cat\'olica de Chile,
  Santiago, Chile.}
\email{gidenittis@mat.uc.cl}

\author[K. Gomi]{Kyonori Gomi}

\address[K. Gomi]{Department of Mathematics, Tokyo Institute of Technology,
2-12-1 Ookayama, Meguro-ku, Tokyo, 152-8551, Japan.}
\email{kgomi@math.titech.ac.jp}

\vspace{2mm}

\date{\today}

\begin{abstract}
This work concerns with the description of  the topological phases of band
insulators of class DIII by using the equivariant cohomology. The main result is the definition of a cohomology class  for general systems of class DIII which  generalizes the   well-known $\Z_2$-invariant given by the Teo-Kane formula in the one-dimension case. In the two-dimensional case this cohomology invariant  allows a complete description of the strong and weak phases. The relation with the KR-theory, the Noether-Fredholm index and the classification of \virg{Real} gerbes are also discussed.
\medskip


\end{abstract}

\maketitle

\tableofcontents

\section{Introduction}\label{sect:intro}

This work concerns the classification of the protected phases of topological insulators of class DIII from the point of view of the equivariant cohomology theory.  In this sense this work complements the research carried out by the same authors for classes AI, AII and AIII \cite{denittis-gomi-14,denittis-gomi-14-gen,denittis-gomi-18-I,denittis-gomi-18-II,denittis-gomi-18-III}. The literature on \emph{topological insulators} is too vast to be reported here without crucial omissions. For this reason, and for the benefit of the reader, we will refer directly to the recent review \cite{chiu-teo-schnyder-ryu-16} which provides a global overview of the subject accompanied by a rich list of references. The works \cite{ryu-schnyder-furusaki-ludwig-10,teo-kane-10,budich-ardonne-13}, among others,
deserve to be mentioned for the specific focus on 
 systems of class DIII.

\medskip

In order to present the logic and the main results of this work we need first to introduce some basic definition.
An \emph{involutive space} is a pair $(X,\tau)$ made of
 a compact Hausdorff space 
 $X$ endowed with a
continuous  \emph{involution}
$\tau:X\to X$, $\tau\circ \tau={\rm Id}_X$. 
The \emph{fixed point set} of $(X,\tau)$ is defined as 
$X^\tau:=\{x\in X\;|\; \tau(x)=x\}$.
A special example is provided by the pair $(\n{S}^1,\iota)$
 made by the unit circle  $\n{S}^1:=\R/2\pi\Z$,  parametrized by $k\in[0,2\pi)$, and endowed with the involution $\iota$ defined by $\iota(k):=-k$. In this specific case the fixed-point set is  $(\n{S}^1)^\iota=\{0,\pi\}$.
 
\medskip 
 
Let
$\s{H}$ be a finite dimensional Hilbert space, $\rr{Herm}(\s{H})$
the set of self-adjoint (or \emph{Hermitian}) operators on $\s{H}$
and $\rr{Herm}(\s{H})^\times$ the subset of invertible elements of $\rr{Herm}(\s{H})$. Let $T:\s{H}\to\s{H}$ and 
$C:\s{H}\to\s{H}$ be two anti-unitary maps that meet the conditions
\begin{equation}\label{eq:int_010}
T^2\;=\;{-}{\bf 1}\;,\qquad C^2\;=\;{+}{\bf 1}\;,\qquad TC=-CT\;
\end{equation}
where ${\bf 1}$ is the identity operator on $\s{H}$.  In the theory of {topological insulators} it is common  to  refer to $T$ as \emph{time-reversal symmetry} (TSR) and to $C$ as  \emph{particle-hole symmetry} (PHS). The conditions on the square say that $T$ is an \emph{odd} symmetry while $C$ is an \emph{even} symmetry. The anti-commutation condition 
corresponds to fixing an arbitrary relative phase between $T$ and $C$. With this convention, the linear operator $\chi:=TC=-CT$, called \emph{chiral symmetry}, satisfies
\begin{equation}\label{eq:int_02}
\chi^2\;=\;{\bf 1}\;,\qquad T\chi=-\chi T\;,\qquad C\chi=-\chi C\;.
\end{equation}
The next definition is borrowed from the standard jargon for topological insulators  \cite{altland-zirnbauer-97,schnyder-ryu-furusaki-ludwig-08,kitaev-09,ryu-schnyder-furusaki-ludwig-10}.
\begin{definition}[Band insulator of class DIII]\label{def:int_01}
A \emph{band insulator of class DIII} over the involutive space $(X,\tau)$ is 
a continuous map 
$$
H\;:\;{X}\;\longrightarrow\; \rr{Herm}(\s{H})^\times
$$
such that
\begin{equation}\label{eq:int_03}
\left\{
\begin{aligned}
C\;H(x)\;&=\;-H(\tau(x))\;C\\
T\;H(x)\;&=\;\phantom{-}H(\tau(x))\;T
\end{aligned}
\right.\;,\qquad \forall\; x\in X\;.
\end{equation}
Two  band insulators of class DIII $H$ and $H'$
 over the same involutive space $(X,\tau)$ and subjected to the same system of symmetries $T,C$ are \emph{unitarily equivalent} if and only if there exists a continuous map $V:X\to\n{U}(\s{H})$ such that
\begin{equation}
\left\{
\begin{aligned}
H'(x)\;&=\;V(x)\;H(x)\;V(x)^*\\
C\;V(x)\;&=\;V\big(\tau(x)\big)\;C\\
T\;V(x)\;&=\;V\big(\tau(x)\big)\;T\\
\end{aligned}
\right.
\;,\qquad \forall\; x\in X\;.
\end{equation}
The unitary equivalence is called \emph{strong} if and only if 
$V(\tau(x))=V(x)$ for all $x\in X$.
When the involutive space is explicitly given by
$({\n{S}}^1,\iota)$, then one refers to a
 \emph{one-dimensional} (1D) band insulators of class DIII.\end{definition}

\medskip

Definition \ref{def:int_01} along with the relations \eqref{eq:int_02} imply 
\begin{equation}\label{eq:int_04}
\chi\;H(x)\;=\;-H(x)\;\chi\;,\qquad \forall\; x\in X\;.
\end{equation}
An immediate consequence of this relation is that the dimension of the space $\s{H}$ must be even, \ie $\rm{dim}(\s{H})=2m$ for some $m\in\N$.

\medskip

Let us focus for the moment on the one-dimensional case.
Let $H$ and $H'$ be two 1-dimensional band insulators of class DIII and consider the \emph{spectrally flattened} band operators
$Q_H:=H |H|^{-1}$ and $Q_{H'}:=H' |H'|^{-1}$. It is known  that the triple $({\n{S}}^1\times\s{H}, Q_{H},Q_{H'})$
defines an element of the KR-group $KR^{-3}({\n{S}}^1,\iota)\simeq\Z_2$ \cite{atiyah-66,kitaev-09,freed-moore-13} (see also Section \ref{sec:KR-theo_II}). Thus, there exists a $\Z_2$-invariant which classifies (the difference of)  1-dimensional band insulators of
class DIII. The aim of this work is to interpret this invariant from a cohomology point of view,
and possibly to generalize it to more generic situations. 

\medskip

The key (not new) observation is that  as a consequence of the commutation relations between $H$ and
the symmetries $T$ and $C$, the topological information carried by $H$ is completely encoded into a continuous map $q_H:{\n{S}}^1\to\n{U}(m)$ such that
\begin{equation}\label{eq:int_05}
q_H\big(\iota(k)\big)\;=\;-q_H(k)^{\mathtt{t}}\;,\qquad \forall\;k\in {\n{S}}^1\;,
\end{equation}
where the superscript $\mathtt{t}$ denotes the \emph{transpose} of the matrix $q_H(k)$. There are at least two immediate ways in which the matrix $q_H$ can produce a topological invariant.
\begin{itemize}
\item[(I)]
The map $q_H$ can be interpreted as the \emph{sewing matrix}
of a \virg{Quaternionic} vector bundle \cite{kahn-59,dupont-69,denittis-gomi-14-gen,denittis-gomi-18-I,denittis-gomi-18-II}. More precisely, the map $q_H$  endows the product bundle ${\n{S}}^1\times\C^m$ with the structure of a \virg{Quaternionic} vector bundle $\bb{E}_{q_H}$ explicitly given by 
\begin{equation}\label{eq:int_06}
(k,{\rm v})\;\longmapsto\;\left(\iota(k), q_H(k)\cdot \overline{\rm v}\right)\;,\qquad \forall\;(k,{\rm v})\in {\n{S}}^1\times\C^m\;.
\end{equation}
Unfortunately, it is known that
\virg{Quaternionic} vector bundles on $({\n{S}}^1,\iota)$ are automatically trivial  \cite[Theorem 1.2]{denittis-gomi-14-gen}. Therefore, this is not this vector bundle the good candidate to detect  the topology of  1-dimensional insulators of class DIII. 

\vspace{1mm}

\item[(II)] In the presence of a fixed point $k = \iota(k)$, the rank $m$ of  $\n{U}(m)$ is forced to be even, \ie $m=2n$. In turn, the determinant of $q_H$ provides a map $\det[ q_H] : {\n{S}}^1 \to \n{U}(1)$ which is \emph{invariant} in the sense that 
$$
\det [q_H(\iota(k))]\; =\; \det[ q_H(k)]\;,\qquad \forall\;k\in {\n{S}}^1
\;.
$$ 
Since the homotopy classes of invariant maps on $({\n{S}}^1,\iota)$ are classified by the group $H^1_{\Z_2}(\n{S}^1,\Z)$ (\cf Appendix \ref{sec:eq:cohom_homot}), one  gets the topological invariant  
$$
\big[\det[ q_H]\big] \;\in\; H^1_{\Z_2}(\n{S}^1,\Z)
$$ 
associated to $H$. However, also in this case the invariant turns out to be trivial since $H^1_{\Z_2}(\n{S}^1,\Z)=0$ \cite[Lemma 5.6]{denittis-gomi-14}. The implication of this fact   will be discussed in Remark \ref{rk:obos}.
\end{itemize}

\medskip

Thus, to detect the class of $KR^{-3}({\n{S}}^1,\iota)$ 
associated to a {one-dimensional} band insulator of class DIII we need to extract from $q_H$
some other type of topological invariant finer than (I) and (II). It turns out that a closer inspection to  the triviality of the \virg{Quaternionic} vector bundle $\bb{E}_{q_H}$ 
allows to build a suitable non-trivial $\Z/2$-invariant. 
The principal result  of this work is  the construction of such a topological invariant for band insulators of class DIII over a  large class of involutive spaces $(X,\tau)$.

\medskip

More concretely,  given a sewing matrix $q : X \to \n{U}(2n)$ on the  involutive space $(X,\tau)$ subject to certain conditions  (\cf Assumption \ref{assumption}), one can associate to $q$ a homotopy invariant 
\begin{equation}\label{eq:int_01}
\nu_{q} \;\in\; H^1(X,\Z)/H^1_{\Z_2}\big(X| X^\tau,\Z(1)\big)\;.
\end{equation}
In  formula \eqref{eq:int_01}, $H^1(X,\Z)$ is the first  cohomology group of $X$ with integer coefficients while  $H^1_{\Z_2}(X| X^\tau,\Z(1))$ is the first equivariant  cohomology group $X$ with  coefficients in the local system $\Z(1)$ relative to the fixed point set $X^\tau$ \cite[Section 3.1]{denittis-gomi-14-gen}. 
In the case of the involutive circle $(\n{S}^1,\iota)$,
one has that 
$$
H^1(\n{S}^1,\Z)\;\simeq\;\Z\;\simeq\;H^1_{\Z_2}\big(\n{S}^1| (\n{S}^1)^\iota,\Z(1)\big)\;,$$
but the quotient group in \eqref{eq:int_01} turns out to be isomorphic to $\Z_2$ (\cf Lemma \ref{lemma:1D-range}), showing that $\nu_{q}$ is 
 a $\Z_2$-invariant as expected.
 
 \medskip 

It is worth to point out, however, that the invariant $\nu_q$ of the sewing matrix $q$ cannot serve as an \emph{absolute} invariant for band insulators   of class DIII. This is a consequence of the fact that the way to identify the symmetries with a fixed \emph{standard form} is not unique, and can be altered by an automorphism (\cf Proposition \ref{prop:st_rep}). Accordingly, 
the sewing matrix associated with $H$ is not unique, and if 
$q_H$ and  $q'_H$ are two sewing matrices associated with the same $H$
one can  have $\nu_{q_H} \neq \nu_{q'_H}$ in general. In this sense, $\nu_{q_H}$ does not produce an invariant of the Hamiltonian $H$. However, if one fixes a way to identify the symmetries with the standard form, then the sewing matrix can be used to detect the difference of two band insulators $H$ and $H'$ of class DIII . 
In this sense the sewing matrix can be used to construct a \emph{relative} invariant which provides a 
a well-defined non-trivial homomorphism 
\begin{equation}\label{int_nu}
 KR^{-3}(X,\tau)\;\supseteq\;{\rm Ker}(\kappa)\; \stackrel{\nu}{\longrightarrow}\; H^1(X,\Z)/H^1_{\Z_2}\big(X, X^\tau,\Z(1)\big)\;
\end{equation}
 (\cf Theorem \ref{thm:homomorphism_from_KR} for more details). In the one-dimensional case $(\n{S}^1,\iota)$
this construction provides a bijection $\nu : KR^{-3}(\n{S}^1,\iota) \simeq \Z_2$ (Proposition \ref{prop:iso-one}) and the $\Z_2$-value of the class $\nu_q$ associated with the sewing matrix $q$ can be computed with the  well known \emph{Teo-Kane formula} (Theorem \ref{thm:invariant_by_Pfaffian}) or as the \emph{$\Z_2$-index} of an associated Toeplitz operator (Theorem \ref{the:index}).

\subsection*{Structure of the paper.}
 In {\bf Section \ref{sec:Z2_inv}}, we start by fixing the standard form of the symmetries of a band insulator of class DIII. Then, we introduce the invariant $\nu_q$ of a sewing matrix $q$ in a general setup.  The construction of the homomorphism $\nu$ in equation \eqref{int_nu} and its relation with the KR-theory are also provided in this section. 
{\bf Section \ref{sec:one-dim-case}} is devoted to the study of the one-dimensional case. In this section we provide the equivalence of $\nu_q$ with the well-known Teo-Kane formula \cite[eq. (3.70)]{chiu-teo-schnyder-ryu-16},
 we clarify the role of $\nu_q$ as an obstruction class
 and 
  we establish a bulk-edge-type correspondence
  as the coincidence of $\nu_q$ with a $\Z_2$-index for Toeplitz  operators. Finally, we will show that the homomorphism \eqref{int_nu} amounts to an isomorphism in the one-dimensional case and we will shed some light on the role
 of this invariant in the classification of \virg{Real} gerbes over $(\n{S}^1,\iota)$. In {\bf Section \ref{sec:two-dim-case}} we investigate the application of the invariant $\nu_q$ for two-dimensional  band insulators of class DIII over the sphere (Dirac case) and the torus (Bloch case). It turns out that $\nu_q$ provides a complete description of the weak topological invariants in the two-dimensional case.
 {\bf Appendix~\ref{sec:eq:cohom_homot}} contains a brief summary of the relationship between cohomology and homotopy in the $\Z_2$-equivariant category
 and {\bf Appendix~\ref{sec:Z2_inv-equiv}} deals with the notion of strong equivalence between \virg{Quaternionic} 
structures.

\medskip

\subsection*{Acknowledgements}
GD's research is supported by the grant {Fondecyt Regular - 1190204}. KG’s research is supported by JSPS KAKENHI
Grant Numbers 20K03606 \& JP17H06461.

\section{Construction of the   $\Z_2$-invariant}\label{sec:Z2_inv}

\subsection{Underlying \virg{Quaternionic} 
structure}\label{sec:Z2_inv-sewing}
In view of  condition \eqref{eq:int_04}, and the invertibility of $H(x)$, we argued that
band insulators of class DIII can be seated only in Hilbert spaces $\s{H}$ of even dimension. Moreover, the eigenspaces of the symmetry $\chi$, determined by the range of the eigenprojections $\chi_\pm:=\frac{1}{2}(1\pm\chi)$, must have  same dimension. 
As a consequence there exists a (non unique) unitary isomorphism $\s{H}\simeq\C^{2m}$ which provides the matrix representation
\begin{equation}\label{eq:sew_01}
\chi\;=\;
\left(\begin{array}{cc}+{\bf 1}_m & 0 \\
0 & -{\bf 1}_m\end{array}\right)
\end{equation}
where ${\bf 1}_m$ is the $m\times m$ identity matrix. For the next result we need to introduce the operator $K$ which implements the standard real structure on $\C^m$. This is nothing more than the complex conjugation $K {\rm v}:=\overline{\rm v}$ for every ${\rm v}\in\C^m$.
\begin{proposition}(Standard representation)\label{prop:st_rep}
Let $H,T,C$ be the elements which describe a band insulator of class DIII over the involutive space $(X,\tau)$, according to Definition \ref{def:int_01}. Then, there exists a (non unique) unitary isomorphism $\s{H}\simeq\C^{2m}$
which provides the matrix representation
\begin{equation}\label{eq:sew_02}
T\;=\;
\left(\begin{array}{cc}0 & - K \\
+ K & 0\end{array}\right)\;,\qquad
C\;=\;
\left(\begin{array}{cc}0 & - K \\
- K & 0\end{array}\right)\;.
\end{equation}
Moreover, the Hamiltonian $H:{X}\to \rr{Herm}(\s{H})^\times$ is represented by
\begin{equation}\label{eq:sew_04}
H(x)\;=\;
\left(\begin{array}{cc}0 & h(x)^{-1}\\
h(x) & 0\end{array}\right)\;,
\end{equation}
where the map $h:{X}\to\n{GL}(m)$ meets the condition
\begin{equation}\label{eq:sew_05M}
h\big(\tau(x)\big)\;=\;-K\;h(x)^{-1}\;K\;,\qquad\forall\; x\in X\;.
\end{equation}
\end{proposition}
\proof
Under the unitary transformation $\s{H}\simeq\C^{2m}$ which provides the representation \eqref{eq:sew_01}, the operator $T$
 takes the form
$$
T\;=\;U_T
\left(\begin{array}{cc}K & 0 \\
0 & K\end{array}\right)
$$
where $U_T$ is a unitary operator. Since the anti-linear part induced by $K$ commutes with $\chi$, it follows that the second constraint in \eqref{eq:int_02} can be  satisfied if and only if
$\chi U_T=-U_T\chi$. However, this forces  $U_T$ to be off-diagonal with respect to the diagonal representation of $\chi$, and in turn 
$$
T\;=\;
\left(\begin{array}{cc}0 & \widetilde{U}_T K \\
\widetilde{U}_T^* K & 0\end{array}\right)
$$
with $\widetilde{U}_T\in\n{U}(m)$ a given $m\times m$ unitary matrix. The condition $T^2=-{\bf 1}$ implies that $K\widetilde{U}_T K=-\widetilde{U}_T$. Therefore, one gets
$$
\left(\begin{array}{cc}
{\bf 1}_m &  0 \\
0 & \widetilde{U}_T\end{array}\right)
\left(\begin{array}{cc}0 & \widetilde{U}_T K \\
\widetilde{U}_T^* K & 0\end{array}\right)
\left(\begin{array}{cc}{\bf 1}_m &  0 \\
0 & \widetilde{U}_T^*\end{array}\right)
\;=\;\left(\begin{array}{cc}0 & - K \\
+ K & 0\end{array}\right)
$$
and at the same time
$$
\left(\begin{array}{cc}
{\bf 1}_m &  0 \\
0 & \widetilde{U}_T\end{array}\right)\chi
\left(\begin{array}{cc}{\bf 1}_m &  0 \\
0 & \widetilde{U}_T^*\end{array}\right)
\;=\;\chi\;.
$$
Then, up to a modification of the initial unitary transform  $\s{H}\simeq\C^{2m}$, one can represent $T$ as in \eqref{eq:sew_02} without affecting the  representation \eqref{eq:sew_01}. The representation of $C$ in \eqref{eq:sew_02} is obtained by the formula $C=\chi T$. The anti-commutation relation between  $H(x)$ and $\chi$ implies that  $H(x)$  must be off-diagonal. The invertibility of $H(x)$ for every $x\in X$  justifies the representation \eqref{eq:sew_04}. Condition \eqref{eq:sew_05M} follows from the first of \eqref{eq:int_03}.
\qed

\medskip

A straightforward adaptation of the arguments used in the proof of
 Proposition \ref{prop:st_rep} provides the  following result that will be used several times in the following.

\begin{lemma}\label{lemma:diag_form}
Let $V\in\n{U}(2n)$ be a unitary matrix which commutes with 
$\chi$ and $T$   as given by the representations \eqref{eq:sew_01} and \eqref{eq:sew_02}, respectively. Then there exists a unitary matrix $\phi_V\in \n{U}(m)$ such that
$$
V\;=\;\left(\begin{array}{cc} \phi_V & 0 \\
0 & \overline{\phi_V}\end{array}\right)\;,\qquad \overline{\phi_V}\;:=\;K\phi_VK\;.
$$
\end{lemma}

\medskip

Given the invertible matrix  $H\in \rr{Herm}(\s{H})^\times$ one can define the \emph{flattened} operator
\begin{equation}\label{eq:sew_01_01-1}
Q_H\;:=\;\frac{H}{|H|}\;,\qquad |H|\;:=\;\sqrt{H^* H}\;.
\end{equation}
Evidently $Q_H\in \n{U}(\s{H})\cap \rr{Herm}(\s{H})^\times$ is a unitary and self-adjoint operator, namely an involution $Q_H^2={\bf 1}$.
If $H: {X}\to \rr{Herm}(\s{H})^\times$ is a band insulator of class DIII, then $Q_H:{X}\to \rr{Herm}(\s{H})^\times$ is also a 
 band insulator of class DIII, and Proposition \ref{prop:st_rep} applies. In particular one has that
\begin{equation}\label{eq:sew_01_01}
Q_H(x)\;=\;\left(\begin{array}{cc}0 & q_H(x)^{-1}\\
q_H(x) & 0\end{array}\right)\;=\;
\left(\begin{array}{cc}0 & q_H(x)^{*}\\
q_H(x) & 0\end{array}\right)\;,
\end{equation}
where the map $q_H:{X}\to\n{U}(n)$ meets the condition
\begin{equation}\label{eq:sew_05}
q_H\big(\tau(x)\big)\;=\;-Kq_H(x)^{\ast}K\;=\;-q_H(x)^\mathtt{t}\;.
\end{equation}
The map $q_H$ allows to endow the trivial product bundle $X\times \C^m$ with the extra structure $\Theta_H:X\times \C^m\to X\times \C^m$ defined by
\begin{equation}\label{eq:sew_01_02}
\Theta_H\;:\; (x,{\rm v})\;\longmapsto\;\left(\tau(x), q_H(k)\cdot \overline{\rm v}\right)\;,\qquad \forall\; (x,{\rm v})\in X\times\C^m\;.
\end{equation}
It turns out that $(X\times \C^m, \Theta_H)$ is a \virg{Quaternionic} vector bundle in the sense of \cite{kahn-59,dupont-69,denittis-gomi-14-gen,denittis-gomi-18-I,denittis-gomi-18-II}. 
\begin{definition}[Sewing matrix]\label{def:sew_mat}
Let $H:{X}\to \rr{Herm}(\s{H})^\times$ be a band insulator of class DIII over the involutive space $(X,\tau)$ according to Definition \ref{def:int_01}. Then $(X\times \C^m, \Theta_H)$ will be called the \emph{associated} \virg{Quaternionic} vector bundle, and $q_H$  the 
\emph{sewing matrix} of the associated \virg{Quaternionic} structure.
\end{definition}

\medskip

It is known that the existence of fixed points of the involutive space $(X,\tau)$ implies that the rank of the \virg{Quaternionic} vector bundles over $(X,\tau)$ must be even \cite[Proposition 2.1]{denittis-gomi-14-gen}. This fact has an immediate consequence:
\begin{proposition}\label{prop:4-dim}
Let $H:{X}\to \rr{Herm}(\s{H})^\times$ be a band insulator of class DIII over the involutive space $(X,\tau)$, according to Definition \ref{def:int_01}. Assume that $X$ is path-connected and 
$X^\tau\neq\emptyset$. Then the associated \virg{Quaternionic} vector bundle has even rank $m =2n$ and $q_H:X\to\n{U}(2n)$. Moreover, the standard representation described in  Proposition \ref{prop:st_rep} is induced by a unitary isomorphism 
 $\s{H}\simeq\C^{4n}$.
\end{proposition}
%

%------%
\subsection{The topological invariant associated to the sewing matrix}\label{sec:top_sew_mat}
In order to define a sufficiently good topological invariant
for band insulators of class DIII we need to establish some working hypotheses.

\begin{assumption}\label{assumption}
Let $(X,\tau)$ be  an involution space and $q:X\to \n{U}(2n)$ a 
sewing matrix, \ie a map such that
$$
q\big(\tau(x)\big)\;=\;-q(x)^\mathtt{t}\;,\qquad \forall\; x\in X\;.
$$
Let $\bb{E}_q=X\times \C^{2n}$ be the \virg{Quaternionic} vector bundle with \virg{Quaternionic} $\Theta_q$ structure induced by $q$, \ie
$$
\Theta_q\;:\; (x,{\rm v})\;\longmapsto\;\left(\tau(x), q(k)\cdot \overline{\rm v}\right)\;,\qquad (x,{\rm v})\in X\times\C^{2n}\;.
$$
We will assume that:
\begin{itemize}
\item[(a)]
$X$ is a path connected $\Z_2$-CW complex;
\vspace{1mm}
\item[(b)]
The fixed point set is non-empty, \ie $X^\tau\neq\emptyset$;
\vspace{1mm}

\item[(c)]
The FKMM invariant $\kappa(\bb{E}_q) \in H^2_{\Z_2}(X| X^\tau,\Z(1))$ is trivial.
\end{itemize}
\end{assumption}

\medskip

For the definition and the properties of the \emph{FKMM invariant} for  \virg{Quaternionic} vector bundles we refer to \cite{denittis-gomi-14-gen,denittis-gomi-18-I,denittis-gomi-18-II}. 
For more details on the relative cohomology groups $H^\bullet_{\Z_2}(X| X^\tau,\Z(1))$ we will refer to \cite[Section 3.1]{denittis-gomi-14-gen} and references therein.
It is  worth recalling
 that these groups fit into the long exact sequence for the pair $(X, X^\tau)$. The relevant part of the latter for the aims of this work  reads
$$
\begin{aligned}
0\; &\to\;
H^1_{\Z_2}\big(X| X^\tau,\Z(1)\big)\; \to\;
H^1_{\Z_2}\big(X,\Z(1)\big)\;\to\;H^1_{\Z_2}(X^\tau,\Z(1))\; \to\;\cdots\\
\cdots\;&\to \;
H^2_{\Z_2}\big(X|X^\tau,\Z(1)\big)\; \to\;
H^2_{\Z_2}\big(X,\Z(1)\big)\;\to\;\cdots\;
\end{aligned}
$$
where the first term of the sequence corresponds to $H^0_{\Z_2}(X^\tau,\Z(1))=0$ \cite{gomi-15}.
As a matter of fact, the image of $\kappa(\bb{E}_q)$ under $H^2_{\Z_2}(X|X^\tau,\Z(1))\to
H^2_{\Z_2}(X,\Z(1))$ is the \emph{\virg{Real} Chern class} $c_1^R(\det (\bb{E}_q))$ that classifies the determinant \virg{Real} line bundle $\det (\bb{E}_q)$ \cite{kahn-59,denittis-gomi-14}.
 Therefore,  $\kappa(\bb{E}_q) = 0$ implies that the determinant line bundle $\det (\bb{E}_q)$ is trivial as \virg{Real} line bundle. 
Let us also recall the  isomorphisms  $H^1(X,\Z) \simeq [X, \n{U}(1)]$ and $H^1_{\Z_2}(X,\Z(1)) \simeq [X, \n{U}(1)]_{\Z_2}$, where the (standard) involution on $\n{U}(1)$ is given by the complex conjugation $z \mapsto \overline{z}$ \cite[Proposition A.2]{gomi-15}.

\medskip

For the next result we need to introduce some notation. Let $C(X, \n{U}(1))$ be the set of continuous functions from $X$ to $\n{U}(1)$ and $q:X\to \n{U}(2n)$ a 
sewing matrix as in Assumption \ref{assumption}. When $x\in X^\tau$ is a  fixed point under the involution, then $q(x)=-q(x)^\mathtt{t}$ is a skew-symmetric matrix and one can calculate the associated \emph{Pfaffian} ${\rm Pf}[q(x)]\in\n{U}(1)$.
Let us define the subset 
\begin{equation*}
C(X, \n{U}(1))_q \;:=\;
\left\{
p : X \to \n{U}(1) \left|\; 
\begin{aligned}
\det [q(x)] &= p(\tau(x)) p(x)&& \forall\;x \in X, \\
p(x) &= {\rm Pf}[q(x)] && \forall\;x \in X^\tau
\end{aligned}
\right\}\right.\;.
\end{equation*}
Under the isomorphism $H^1(X,\Z)\simeq [X, \n{U}(1)]$ the subset 
$C(X, \n{U}(1))_q $ identifies 
a subset $H^1(X,\Z)_q\subset H^1(X,\Z)$ which coincides with the set of homotopy classes of maps in $C(X, \n{U}(1))_q$.
\begin{lemma}\label{lemma_01}
Under Assumption \ref{assumption}, the following facts hold true:
\begin{itemize}
\item[(1)]
The set $C(X, \n{U}(1))_q$ is  non-empty, and is a torsor under the group of $\Z_2$-equivariant maps $r : X \to \n{U}(1)$ which take the value $1$ on the fixed point set $X^\tau$.
\vspace{1mm}
\item[(2)]
The set $H^1(X,\Z)_q $ is a torsor under the group $H^1_{\Z_2}(X| X^\tau,\Z(1))$ with action  induced by 
$$
H^1_{\Z_2}\big(X| X^\tau,\Z(1)\big)\; \to\; H^1_{\Z_2}\big(X,\Z(1)\big)\; \to\; H^1(X,\Z)\;\supset\;H^1(X,\Z)_q\;.
$$
\end{itemize}
\end{lemma}
\proof
Let us start with (1).
As pointed out above, the triviality of the FKMM invariant $\kappa(\bb{E}_q)$ implies the triviality of the determinant \virg{Real} line bundle $\det(\bb{E}_q)$. 
Since the information on $\det(\bb{E}_q)$ is completely encoded in $\det [q] : X \to \n{U}(1)$, from the triviality of $\det(\bb{E}_q)$ one infers the existence of the a global  trivialization  $p : X \to \n{U}(1)$ which intertwines between the \virg{Real} structure provided by $\det [q]$ and the \emph{trivial} (or constant) \virg{Real} structure. The latter fact is equivalent to the constraint
$\det [q(x)] = p(\tau(x)) p(x)$ for all $x \in X$. 
From \cite[Proposition 2.10]{denittis-gomi-18-I} one gets that
  the FKMM invariant $\kappa(\bb{E}_q) \in H^2_{\Z_2}(X| X^\tau,\Z(1))$ coincides with the injective image under 
  $$
  H^1_{\Z_2}\big(X^\tau,\Z(1)\big)/H^1_{\Z_2}\big(X,\Z(1)\big)\; \longrightarrow\; H^2_{\Z_2}\big(X| X^\tau,\Z(1)\big)
  $$ 
  of the element whose representative is given by the map $X^\tau\ni x\mapsto p(x)/{\rm Pf}[q(x)]\in \Z_2$ where $p$ is the global trivialization introduced above.  Assumption \ref{assumption} (c) says that  $\kappa(\bb{E}_q)$ is trivial. Hence, one can choose $p$ so that $p(x) = {\rm Pf}[q(x)]$ for every $x \in X^\tau$. This proves that the set $C(X, \n{U}(1))_q$ is non-empty. For a given pair $p', p \in C(X, \n{U}(1))_q$, let $r : X \to \n{U}(1)$
be the map  $r(x) := p'(x) \overline{p(x)}$. It turns out that $r(\tau(x)) = \overline{r(x)}$ for every $x \in X$ and $r(x) = 1$ when $x \in X^\tau$. Hence $C(X, \n{U}(1))_q$ is a torsor as stated.
For (2) it is enough to consider the  the homotopy classes of maps in (1).
\qed

\medskip

Let us recall that by definition  a torsor coincides with the orbit generated by any of its points.   
We are now in position to introduce the main topological invariant 
of this work.

\begin{definition}[DIII-invariant] \label{dfn:invariant}
Under Assumption \ref{assumption}, one defines $\nu_q$ to be the orbit of the subset $H^1(X,\Z)_q \subset H^1(X,\Z)$ under $H^1_{\Z_2}(X| X^\tau,\Z(1))$, \ie
$$
\nu_q\;: =\; \big[H^1(X,\Z)_q\big]\; \in\; H^1(X,\Z)/H^1_{\Z_2}\big(X| X^\tau,\Z(1)\big)\;.
$$
\end{definition}

\medskip

Let us point out that the notation used in Definition \ref{dfn:invariant} is a little redundant since $[H^1(X,\Z)_q]$ and $H^1(X,\Z)_q$ coincide as a set in view of the torsor property.
Anyway, the notation proposed as the advantage of emphasizing that
$\nu_q$ is an equivalence class (or an orbit).
The topological properties of $\nu_q$ are specified in the following result.

\begin{theorem} \label{thm:invariant}
Under Assumption \ref{assumption} the following facts hold  true:
\begin{itemize}
\item[(1)]
$\nu_q$ is an invariant of the homotopy class of $q$.

\vspace{1mm}

\item[(2)]
Let $q_0 : X \to \n{U}(2n)$ be the constant map given by
$$
q_0(x)\;=\;Q\;: =\;
\left(
\begin{array}{rr}
0 & -{\bf 1}_n \\
+ {\bf 1}_n & 0
\end{array}
\right)\;,
$$
then $\nu_{q_0} = 0$.

\vspace{1mm}

\item[(3)]
Let $q : X \to \n{U}(2n)$ and $q' : X \to \n{U}(2n')$ be as in Assumption \ref{assumption}, and consider $q \oplus q' : X \to \n{U}(2n+2n')$. Then 
$$
\nu_{q \oplus q'}\; =\; \nu_q\; +\; \nu_{q'}
$$
as elements of the quotient group $H^1(X,\Z)/H^1_{\Z_2}(X| X^\tau,\Z(1))$.

\vspace{1mm}

\item[(4)]
Let $q : X \to \n{U}(2n)$ and $q' : X \to \n{U}(2n)$ be related by
$$
q'(x)\; =\; h(\tau(x))^\mathtt{t}\; q(x)\; h(x)
$$
for a given map $h : X \to \n{U}(2n)$. Then, 
$$
\nu_{q'}\; =\; \det[h] \cdot \nu_q\;
=\; \big[\det [h] \cdot H^1(X,\Z)_q\big] \;.
$$

\end{itemize}
\end{theorem}
\proof
(1) Let $q$ and $q'$ be sewing matrices satisfying Assumption \ref{assumption}. Suppose that these two matrices are connected by a homotopy $\tilde{q} : X \times [0, 1] \to \n{U}(2n)$ of sewing matrices. The map $\tilde{q}$ defines the \virg{Quaternionic} vector bundles $\bb{E}_{\tilde{q}}$ over $X \times [0, 1]$. By the homotopy property of \virg{Quaternionic} vector bundles \cite[Theorem 2.3]{denittis-gomi-14-gen}, $\bb{E}_{\tilde{q}}$ is isomorphic to the pull-back of $\bb{E}_q$ under the projection $X \times [0, 1] \to X$. Hence, both the determinant \virg{Real} line bundle of $\bb{E}_{\tilde{q}}$ and the FKMM invariant of $\bb{E}_{\tilde{q}}$  are trivial. 
As a result, $C(X \times [0, 1], \n{U}(1))_{\tilde{q}}$ is non-empty as a consequence of Lemma \ref{lemma_01}. Let us choose a $\tilde{p} \in C(X \times [0, 1], \n{U}(1))_{\tilde{q}}$. Such a  $\tilde{p}$ provides a homotopy between the representatives $p:=\tilde{p}|_{X \times \{ 0 \}}$ and $p':=\tilde{p}|_{X \times \{ 1 \}}$ of $\nu_q$ and $\nu_{q'}$, respectively. Therefore, $\nu_q=\nu_{q'}$.\\
(2) One can take $p_0 \in C(X, \n{U}(1))_q$ to be the constant map  $p_0(x) = 1$. This map represents the trivial element $0 \in H^1(X,\Z)$, so that $\nu_{q_0} = 0$. \\
(3) If $p \in C(X, \n{U}(1))_q$ and $p' \in C(X, \n{U}(1))_{q'}$, then the map $pp' : X \to \n{U}(1)$ given by 
$$
(pp')(x)\; :=\; p(x)p'(x)
$$ belongs to $C(X, \n{U}(1))_{q \oplus q'}$. In the quotient group  $H^1(X,\Z)/H^1_{\Z_2}(X| X^\tau,\Z(1))$, the invariants $\nu_q$ and $\nu_{q'}$ are represented by the homotopy classes $[p] \in H^1(X,\Z)$ and $[p'] \in H^1(X,\Z)$, respectively. One has that $[pp'] = [p] + [p'] \in H^1(X,\Z)$, so that $\nu_{q \oplus q'} = \nu_q + \nu_{q'}$. \\
(4) If $p \in C(X, \n{U}(1))_q$, then one gets $p' \in C(X, \n{U}(1))_{q'}$ by setting   
$$p'(x) \;:=\; \det [h(x)]p(x).$$ Thus, the multiplication by $\det[ h]$ induces a transformation on $H^1(X,\Z)$ which carries $H^1(X,\Z)_q$ to $H^1(X,\Z)_{q'}$. As a consequence,  $\nu_q\mapsto \nu_{q'}$ in the quotient group.
\qed

\medskip

\begin{remark}
Let us recall that the for the equivariant cohomology
the concept of reduced cohomology 
makes sense (see \eg \cite[Section 5.1]{denittis-gomi-14}). 
Therefore, one has the usual decomposition
$$
H^1_{\Z_2}\big(X,\Z(1)\big)\; \simeq\; \tilde{H}^1_{\Z_2}\big(X,\Z(1)\big)\;  \oplus\; H^1_{\Z_2}\big(\{\ast\},\Z(1)\big),
$$
where  $\tilde{H}^1_{\Z_2}(X,\Z(1))$ is the reduced cohomology group and  $H^1_{\Z_2}(\{\ast\},\Z(1)) \simeq \Z_2$. By means of a spectral sequence, one can show
$$
\tilde{H}^1_{\Z_2}\big(X,\Z(1)\big)\; \simeq\;
H^1(X,\Z)^{\Z_2}_- \;:=\; 
\{
f \in H^1(X; \Z) |\ \tau^*(f) = -f \}.
$$
In particular, $\tilde{H}^1_{\Z_2}(X,\Z(1))$ turns out to be a free abelian group. Since the generator of $H^1_{\Z_2}(\{\ast\},\Z(1))  \subset H^1_{\Z_2}(X,\Z(1))$ can be represented  by the constant map $X\ni x\mapsto-1\in \n{U}(1)$, it follows from the exact sequence 
$$
0\; \to\;
H^1_{\Z_2}\big(X| X^\tau,\Z(1)\big)\; \to\;
H^1_{\Z_2}\big(X,\Z(1)\big)\;=\;\tilde{H}^1_{\Z_2}\big(X,\Z(1)\big)\;\oplus\;\Z_2
$$
that $H^1_{\Z_2}(X| X^\tau,\Z(1))$ injects into the free part $\tilde{H}^1_{\Z_2}(X,\Z(1))$.
 \hfill $\blacktriangleleft$
\end{remark}

\subsection{The relation with the KR-theory}\label{sec:KR-theo}
The KR-group $KR^{-3}(X,\tau)$ of a space $X$ with involution $\tau$ may be regarded as a group which classifies the \virg{difference} of two band insulators of class DIII. In this section we will relate $KR^{-3}(X,\tau)$ with the invariant $\nu_q$ of a sewing matrix $q$, provided that  certain assumptions are satisfied.

\medskip

Let us start by recalling the formulation of the \emph{twisted equivariant K-theory} \cite{freed-moore-13} (see also \cite[Section 5]{denittis-gomi-19} and \cite{gomi-17}). Let $\n{G}$ be a finite group acting on a space $X$ endowed with two homomorphisms $\phi : \n{G} \to \Z_2$ and $c : \n{G} \to \Z_2$ and a group $2$-cocycle $\sigma : \n{G} \times \n{G} \times X \to \n{U}(1)$, which satisfies
$$
\sigma(g_2, g_3, x)^{\phi(g_1)}\; \sigma(g_1 g_2, g_3, x)^{-1}\;
\sigma(g_1, g_2 g_3, x)\; \sigma(g_1, g_2, g_3(x))^{-1}\; =\; 1\;
$$
for all $g_1,g_2,g_3\in\n{G}$ and all $x\in X$.
A \emph{$(\phi, \sigma, c)$-twisted (ungraded) vector bundle} with a \emph{$Cl_{p, q}$-action} is a (finite-rank) Hermitian vector bundle $\pi : \bb{E} \to X$ equipped with: 
\begin{itemize}
\item[(a)]
An isometric bundle map $\rho(g) : \bb{E} \to \bb{E}$ covering the action of $g \in \n{G}$ on $X$ for each $g \in \n{G}$; 
\vspace{1mm}
\item[(b)]
  Unitary maps $\gamma_1, \ldots, \gamma_{p+q}$ on $\bb{E}$,
which satisfy the Clifford relations
$$
\gamma_i\gamma_j + \gamma_j\gamma_i \;=\;
\left\{
\begin{array}{rll}
2 &&\mbox{if}\;\; i = j = 1, \ldots, p\;, \\
-2 && \mbox{if}\;\; i = j = p+1, \ldots, p+q\;, \\
0 && \mbox{otherwise}\;;
\end{array}
\right.
$$
\vspace{1mm}
\item[(c)]
The compatibility  relations
\begin{align*}
\ii \rho(g) &\;=\; \phi(g) \rho(g) \ii\;, \\
\rho(g)\rho(h) &\;=\; \sigma(g, h) \rho(gh)\;, \\
\gamma_j \rho(g) &\;=\; c(g) \rho(g) \gamma_j\;, 
\end{align*}
valid for every $g,h\in\n{G}$ and $j=1,\ldots,p+q$.
\end{itemize}
A \emph{homomorphism} from a twisted bundle $(\bb{E}, \rho, \gamma)$ to a twisted bundle $(\bb{E}', \rho', \gamma)$ is a map of complex vector bundles $f : \bb{E} \to \bb{E}'$ such that
\begin{align*}
f \circ \rho(g) \;=\; \rho'(g) \circ f\;, \qquad
f \circ \gamma_j \;=\; \gamma'_j \circ f\;,
\end{align*}
for every $g\in\n{G}$ and every $j=1,\ldots,p+q$.
The notion of \emph{isomorphism} follows naturally. 

\medskip

Let $(\bb{E}, \rho, \gamma)$ be a  twisted bundle with a Clifford action. An invertible Hermitian map $H : \bb{E} \to \bb{E}$ is said to be \emph{compatible} with the twisted action if it holds true that
\begin{align*}
H \rho(g) \;=\; c(g) \rho(g) H\;, \qquad
H \gamma_j \;= \;- \gamma_j H.
\end{align*}
for every $g\in\n{G}$ and every $j=1,\ldots,p+q$.
If $H$ and $H'$ are invertible Hermitian maps on $(\bb{E}, \rho, \gamma)$ and $(\bb{E}', \rho', \gamma')$ respectively, such that $f \circ H = H' \circ f$ for a  given isomorphism $f : \bb{E} \to \bb{E}'$, then $H$ and $H'$ will be regarded as \emph{isomorphic}. Note that the self-adjoint involution $Q_H := H/\lvert H \rvert$ is nothing but a gradation, or a $\Z_2$-grading of the twisted bundle $\bb{E}$. Since $Q_H$ and $H$ are homotopic, we can generalize Karoubi's formulation of twisted equivariant $K$-theory \cite{donovan-karoubi-70,rosenberg-89,freed-hopkins-teleman-11,freed-moore-13,gomi-17} as follows:
\begin{definition}[Twisted equivariant K-theory]
Let $X$, $\phi$, $\sigma$ and $c$ be as above.
\begin{itemize}
\item[(a)]
We define ${}^\phi\mathcal{M}^{(\sigma, c) + (p, q)}_{\n{G}}(X)$ to be the monoid of isomorphism classes of triples $(\bb{E}, H_0, H_1)$ consisting of $(\phi, \sigma, c)$-twisted bundles $\bb{E}$ with $Cl_{p, q}$-action and two invertible Hermitian maps $H_0$ and $H_1$ on $\bb{E}$ compatible with the twisted actions.
\vspace{1mm}
\item[(b)]
We define ${}^\phi\mathcal{Z}^{(\sigma, c) + (p, q)}_{\n{G}}(X)$ to be the submonoid of ${}^\phi\mathcal{M}^{(\sigma, c) + (p, q)}_{\n{G}}(X)$ consisting of isomorphism classes of triples $(\bb{E}, H_0, H_1)$ such that $H_0$ and $H_1$ are homotopic within invertible Hermitian maps compatible with the twisted actions.
\vspace{1mm}

\item[(c)]
We define ${}^\phi K^{(\sigma, c) + (p, q)}_{\n{G}}(X) := {}^\phi\mathcal{M}^{(\sigma, c) + (p, q)}_{\n{G}}(X)/{}^\phi\mathcal{M}^{(\sigma, c) + (p, q)}_G(X)$ to be the quotient monoid. 

\end{itemize}
\end{definition}

In view  of the $(1, 1)$-periodicity ${}^\phi K^{(\sigma, c) + (p + 1, q + 1)}_{\n{G}}(X) \cong {}^\phi K^{(\sigma, c) + (p, q)}_{\n{G}}(X)$  \cite{gomi-17}, we put
\begin{equation}\label{diff-ind}
{}^\phi K^{(\sigma, c) + q - p}_{\n{G}}(X) \;: =\; {}^\phi K^{(\sigma, c) + (p, q)}_{\n{G}}(X)\;.
\end{equation}

\medskip

Let us now describe how from a band insulator 
 $H : X \to \rr{Herm}(\n{C}^{2m})^\times$ of class DIII 
on the involutive  space $(X,\tau)$ one gets an invertible Hermitian map on a twisted bundle.
Consider the group  $\n{G} = \Z_2=\{\pm1\}$ as generated by $\tau\equiv -1$, \ie  let $\n{G}$ act on $X$ through the involution. We choose $\phi : \n{G} \to \Z_2$ to be the \emph{identity} homomorphism  $\phi(\pm 1) = \pm1$ 
and $c : \n{G} \to \Z_2$ to be the \emph{trivial} homomorphism  $c(\pm 1) = 0$. These choices can be shortly summarized by $\phi=1$ and $c=0$.
Let  $\sigma : \n{G} \times \n{G} \to \n{U}(1)$ 
be the  $2$-cocycle (independent of the points of $X$)
defined by
\begin{align*}
\sigma(1, 1) \;=\; \sigma(-1, 1) \;=\; \sigma(1, -1) \;=\; 1\;, \qquad
\sigma(-1, -1) \;=\; -1\;.
\end{align*}
Let $\bb{E} = X \times \C^{2m}$ be the product Hermitian vector bundle on $X$ and define $\rho(g) : \bb{E} \to \bb{E}$ as follows:
\begin{equation}\label{eq:KR_01}
\rho(1) \;:\; (x, {\rm v})\; \longmapsto\; (x, {\rm v})\;, \qquad\;
\rho(-1)\;:\; (x, {\rm v}) \; \longmapsto\; (\tau(x), T\bar{{\rm v}})\;
\end{equation}
where $T$ denotes the TRS. 
We also define $\gamma : \bb{E} \to \bb{E}$ to be the unitary map 
defined by
\begin{equation}\label{eq:KR_02}
\gamma \;:\; (x, {\rm v})\; \longmapsto\;  (x, \ii \chi {\rm v})\;.
\end{equation}
where $\chi$ denotes the chiral symmetry.
Then $(\bb{E}, \rho, \gamma)$ is a $(\phi, \sigma)$-twisted vector bundle with $Cl_{0, 1}$-action. Finally, from  $H : X \to \rr{Herm}(\n{C}^{2m})^\times$ one gets the invertible Hermitian map $H : (x, {\rm v}) \mapsto (x, H(x) {\rm v})$ on $\bb{E}$ compatible with the twisted action. As a result one has that differences of band insulators of class DIII provide elements of the twisted equivariant K-theory ${}^\phi K^{\sigma + 1}_{\Z_2}(X)$
where we used the notation \eqref{diff-ind} and the reference to the trivial map $c=0$ has been omitted.
As a matter of fact
a $\phi$-twisted vector bundle is nothing but a \virg{Real} vector bundle, so that ${}^\phi K_{\Z_2}^{0}(X) = KR^0(X,\tau)$ \cite{atiyah-66}. Since the cocycle $\sigma$ has the effect of the degree shift by $-4$ one gets
\begin{equation}\label{eq:k-K}
{}^\phi K^{\sigma + 1}_{\Z_2}(X) \;=\; KR^{-3}(X,\tau)\;.
\end{equation}

\medskip

The next result provides the link between  the group $KR^{-3}(X,\tau)$ which classifies the (differences of) band insulators of class DIII and the  invariant of sewing matrices described in Definition \ref{dfn:invariant}. For that we need to introduce the homomorphism
\begin{equation}\label{eq:FKMM-KR}
\kappa\;:\;KR^{-3}(X,\tau)\; \longrightarrow\; H^1_{\Z_2}\big(X| X^\tau,\Z(1)\big)
\end{equation}
induced by the FKMM invariant \cite{denittis-gomi-14-gen,denittis-gomi-18-I,denittis-gomi-18-II}.
More precisely, the homomorphism $\kappa$ assigns to a band insulator of class DIII, thought of as a class in  $KR^{-3}(X,\tau)$, the FKMM invariant of the associated \virg{Quaternionic} vector bundle as described in Definition \ref{def:sew_mat}.

\begin{theorem} \label{thm:homomorphism_from_KR}
Let $(X,\tau)$ be  a path connected  finite $\Z_2$-CW complex which admits a fixed point.
Then:
\begin{itemize}
\item[(1)] The map \eqref{eq:FKMM-KR} induced by the FKMM invariant of sewing matrices provides a well-defined  homomorphism;

\vspace{1mm}
\item[(2)]  The invariant of sewing matrices introduced in Definition \ref{dfn:invariant} provides a well-defined homomorphism
$$
\nu\; :\; {\rm Ker}(\kappa) \; \longrightarrow\; H^1(X,\Z)/H^1_{\Z_2}\big(X| X^\tau,\Z(1)\big)\;.
$$
\end{itemize}
\end{theorem}
\proof
Let $(\bb{E}, H_0, H_1)$ be a triple representing an element of $KR^{-3}(X,\tau)$. In general, a twisted graded vector bundle on a compact space equivariant under a finite group admits a complementary twisted bundle $\bb{F}$ such that $\bb{E}\oplus \bb{F}\simeq X \times \s{H}$ is isomorphic to a product bundle for some finite dimensional Hilbert space $\s{H}$ and the twisted actions are independent of the points on $X$  (this fact is a straightforward generalization of a property of the usual equivariant vector bundles). Let $H_F:\bb{F}\to\bb{F}$ be an invertible Hermitian map. In view of the definition of $KR^{-3}(X,\tau)$ as a quotient monoid, $(\bb{E}, H_0, H_1)$ and $(\bb{E}\oplus \bb{F}, H_0 \oplus H_F, H_1 \oplus H_F)$ provide the same class in the $K$-theory. Thus, we can assume that $\bb{E} = X \times \s{H}$ with  twisted actions  independent of points on $X$ since  the beginning.
 Therefore, as shown in Proposition \ref{prop:st_rep}, we can represent the twisted actions on $X \times \s{H}$ in a standard form, and express $H_0$ and $H_1$ in terms of sewing matrices $q_0:=q_{H_0}$ and $q_1:=q_{H_1}$, respectively. \\
 (1) With the argument above (and the identification $\s{H}\equiv\C^{4n}$ in view of Proposition \ref{prop:4-dim}) one can specify the action of the map \eqref{eq:FKMM-KR} as follows
\begin{equation}\label{eq:def_k_KR}
 \kappa\big([X\times \C^{4n}, H_0, H_1]\big)\;:=\;\kappa(\bb{E}_{q_0})\;-\;
\kappa(\bb{E}_{q_1}) 
\end{equation}
  where $\bb{E}_{q_i}$ is the \virg{Quaternionic} vector bundle associated with the sewing matrix $q_i$, with $i=0,1$. Notice that this definition is independent of the way to identify the
standard presentation. In fact if one changes the identification, one obtains  new
sewing matrices $q'_i$ which  are related to   $q_i$ via a continuous map $\phi:X\mapsto \n{U}(2n)$ according to 
equation \eqref{eq:inter_01}. This map induces an isomorphism
of the associated \virg{Quaternionic} vector bundles $\bb{E}_{q_i'}\simeq \bb{E}_{q_i}$. 
Since the FKMM
invariant takes the same value on isomorphic \virg{Quaternionic} bundles, it follows that equation \eqref{eq:def_k_KR} is  insensitive on the way of fixing the standard presentation.
Furthermore, if $H_0$ and $H_1$ are homotopic, then so are $q_0$ and $q_1$. In this case the
homotpy invariance of \virg{Quaternionic} bundles imples $\kappa(\bb{E}_{q_0})=
\kappa(\bb{E}_{q_1}) $, and in turn one gets $ \kappa([X\times \C^{4n}, H_0, H_1])=0$
for the trivial class. This fact, along with the  additivity of the
FKMM invariant under the direct sum of \virg{Quaternionic} bundles,
proves that   \eqref{eq:FKMM-KR} is a well-defined homomorphism.\\
 (2)  Let us start by showing that when $[X\times \C^{4n}, H_0, H_1]\in  {\rm Ker}(\kappa)$ then it is possible to choose the sewing matrices $q_i$ in such a way that $\kappa(\bb{E}_{q_i})=0$ for $i=0,1$. To see this observe that $[X\times\C^{4n}, H_0, H_1]\in  {\rm Ker}(\kappa)$ implies 
 $\kappa(\bb{E}_{q_0})=
\kappa(\bb{E}_{q_1}) $. As described in the proof of Lemma \ref{lemma_01}, the 
 FKMM
invariant $\kappa(\bb{E}_{q_i})$ can be represented by the map
 $X^\tau\ni x\mapsto p_i(x)/{\rm Pf}[q_i(x)]\in \Z_2$ where $p_i:X\to \n{U}(1)$ is such that 
 $\det [q_i(x)] = p_i(\tau(x)) p_i(x)$ for all $x \in X$. Let $q_2:X\to\n{U}(2n)$ be the sewing matrix defined by $q_2(x):=q_0(x)^{-1}$. This provides ${\rm Pf}[q_2(x)]=(-1)^n{\rm Pf}[q_0(x)]^{-1}$ for all $x\in X^\tau$. If we define $p_2:X\to \n{U}(1)$  by
 $p_2(x):=(-1)^np_0(x)^{-1}$ for all $x \in X$, then one gets 
  $\det [q_2(x)] = p_2(\tau(x)) p_2(x)$ for all $x \in X$ and the  FKMM
invariant $\kappa(\bb{E}_{q_2})$ is represented by the map
 $$
 X^\tau\;\ni\; x\;\longmapsto \;\frac{p_2(x)}{{\rm Pf}[q_2(x)]}\;=\;\left(\frac{p_0(x)}{{\rm Pf}[q_0(x)]}\right)^{-1}\;\in\;\Z_2\;.
 $$
 This means that $\kappa(\bb{E}_{q_2})=-\kappa(\bb{E}_{q_0})$. Now, let us introduce the 
 Hamiltonian $H_2:{X}\to \rr{Herm}(\C^{4n})^\times$ given by 
$$
H_2(x)\;:=\;
\left(\begin{array}{cc}0 & q_2(x)^{-1}\\
q_2(x) & 0\end{array}\right)\;,\qquad x\in X\;,
$$
and the new triple $(X\times \s{H}', H_0', H_1')$ defined by
$$
\begin{aligned}
(X\times \s{H}', H_0', H_1')\;:&=\;(X\times \C^{4n}, H_0, H_1)\;\oplus\;(X\times \C^{4n}, H_2, H_2)\\
&=\; \big(X\times (\C^{4n}\oplus \C^{4n}), H_0\oplus H_2, H_1\oplus H_2\big)\;.
\end{aligned}
$$
The classes $[X\times \C^{4n}, H_0, H_1]$ and $[X\times \s{H}', H_0', H_1']$ represent the same element in ${\rm Ker}(\kappa)$ and
$$
\begin{aligned}
\kappa(\bb{E}_{q_0'})\;&=\;\kappa(\bb{E}_{q_0}\oplus\bb{E}_{q_2})\;=\;\kappa(\bb{E}_{q_0})+\kappa(\bb{E}_{q_2})\;=\;0\\
\kappa(\bb{E}_{q_1'})\;&=\;\kappa(\bb{E}_{q_1}\oplus\bb{E}_{q_2})\;=\;\kappa(\bb{E}_{q_1})+\kappa(\bb{E}_{q_2})\;=\;\kappa(\bb{E}_{q_0})+\kappa(\bb{E}_{q_2})\;=\;0\;.
\end{aligned}
$$
As a consequence, we can show that
$$
\mathrm{Ker}(\kappa) = {}^\phi\mathcal{N}^{\sigma + (0, 1)}_{\Z_2}(X)/
\left({}^\phi\mathcal{N}^{\sigma + (0, 1)}_{\Z_2}(X) 
\;\cap\; 
{}^\phi\mathcal{Z}^{\sigma + (0, 1)}_{\Z_2}(X)\right)\;,
$$
where the submonoid ${}^\phi\mathcal{N}^{\sigma + (0, 1)}_{\Z_2}(X) \subset {}^\phi\mathcal{M}^{\sigma + (0, 1)}_{\Z_2}(X)$ is defined by
$$
{}^\phi\mathcal{N}^{\sigma + (0, 1)}_{\Z_2}(X)
\;:=\; \left.\left\{ 
(X \times \mathcal{H}, H_0, H_1) \in 
{}^\phi\mathcal{M}^{\sigma + (0, 1)}_{\Z_2}(X)\; \right|\;
\kappa(\mathcal{E}_{q_0}) = \kappa(\mathcal{E}_{q_1}) = 0 
\right\}\;.
$$
Now, let us define the map $\nu$ which associates to  
$(X\times \s{H}, H_0, H_1)\in{}^\phi\mathcal{N}^{\sigma + (0, 1)}_{\Z_2}(X)$ the difference of the invariants of the sewing matrices $q_0$ and $q_1$, \ie
\begin{equation}\label{map_nu}
\nu\big((X\times \s{H}, H_0, H_1)\big)\;:=\;\nu_{q_0} - \nu_{q_1}\;.
\end{equation}
The rest of the proof consists in proving that the map $\nu$ as defined by \eqref{map_nu} provides a well-defined homomorphism
on the group ${\rm Ker}(\kappa)$.
First of all, let us notice that the sewing matrices $q_0,q_1$ depend on the choice of the isomorphism $X \times \s{H} \simeq X \times \C^{4n}$  which makes the twisted actions in the standard form as discussed in Lemma \ref{lemma:diag_form}. By combining the notation introduced in \eqref{eq:KR_01} and \eqref{eq:KR_02} with 
Proposition \ref{prop:st_rep}, one gets that the 
 standard actions are given by
$$
\rho(-1) \;:\; (x, {\rm v})\; \longmapsto\; (x, T{\rm v})\;=\;\left(x, 
\left(
\begin{array}{cc}
0 & -{\bf 1}_{2n} \\
+{\bf 1}_{2n} & 0
\end{array}
\right)
\overline{\rm v}\right)
$$
and 
$$
\gamma \;:\; (x, {\rm v})\; \longmapsto\; (x, \ii \chi {\rm v})\;=\;\left(x, 
\left(
\begin{array}{cc}
  +\ii {\bf 1}_{2n}&0 \\
0&-\ii{\bf 1}_{2n} 
\end{array}
\right)
{\rm v}\right)\;.
$$
Moreover, every continuous map $f : X \to \n{U}(2n)$ provides the automorphism $\psi_f$ of $X \times \C^{4n}$ given by
$$
\psi_f\;:\;(x, {\rm v})\; \longmapsto\;\left(x, 
\left(
\begin{array}{cc}
  f(x)&0 \\
0&\overline{f(\tau(x))} 
\end{array}
\right)
{\rm v}\right)\;.
$$
Let $\phi:X \times \s{H} \cong X \times \C^{4n}$ and
$\phi':X \times \s{H} \cong X \times \C^{4n}$ be two standard representation intertwined by the automorphism $\psi_f$, \ie
$\phi=\psi_f\circ \phi'$. Let $q_0,q_1$ be the sewing matrices associated to $H_0$ and $H_1$, respectively, via the isomorphism 
$\phi$. Similarly, let $q'_0,q'_1$ be the related sewing matrices 
obtained via the isomorphism 
$\phi'$. A straightforward computation shows that 
$$
q'_j(x) \;=\; f(\tau(x))^\mathtt{t}\; q_j(x)\;f(x)\;,\;\qquad\; j=0,1\;.
$$
By using Theorem \ref{thm:invariant} (4) one gets the following fact: If $\nu_{q_j}$ is represented by $p_j \in C(X, \n{U}(1))_{q_j}$, then $\nu_{q'_j}$ is represented by $p_j'=p_j \det [f]$. In terms of classes in $H^1(X,\Z) = [X, U(1)]$, one gets 
$$
\begin{aligned}
{[p_0']} - [p_1']\;&=\;\big[p_0 \det [f]\big] - \big[p_1 \det [f]\big]\\
&=\; \big([p_0] + \big[ \det [f]\big]\big) - \big([p_1] + \big[ \det [f]\big]\big) 
\;=\; [p_0] - [p_1]\;.
\end{aligned}
$$
Therefore $\nu_{q'_0} - \nu_{q'_1} = \nu_{q_0} - \nu_{q_1}$ in $H^1(X,\Z)/H^1_{\Z_2}(X| X^\tau,\Z(1))$, and 
the map $\nu$ defined by \eqref{map_nu} turns out to be 
independent of the choice of the standard representation $X \times \s{H} \simeq X \times \C^{4n}$.
This also implies that  $\nu(X \times \s{H}, H_0, H_1) = \nu(X \times \s{H}', H'_0, H'_1)$ whenever the triples $(X \times \s{H}, H_0, H_1)$ and $(X \times \s{H}', H'_0, H'_1)$ are isomorphic.
  This leads to a map 
  $$
  \nu \;: \;{}^\phi\mathcal{N}^{\sigma + (0, 1)}_{\Z_2}(X)\; \longrightarrow \;H^1(X,\Z)/H^1_{\Z_2}(X| X^\tau,\Z(1))\;,
  $$
  which is a monoid homomorphism by Theorem \ref{thm:invariant} (2) and (3).
 By Theorem \ref{thm:invariant} (1), the monoid homomorphism is trivial on the submonoid ${}^\phi\mathcal{Z}^{\sigma + (0, 1)}_{\Z_2}(X)\cap{}^\phi\mathcal{N}^{\sigma + (0, 1)}_{\Z_2}(X)
 $. Consequently, we get a well-defined homomorphism $\nu$ on the group ${\rm Ker}(\kappa)$. Equation \eqref{eq:k-K} completes the proof. 
\qed

\begin{remark}\label{rk:obos-S1}
It is well known that the involutive space $({\n{S}}^1,\iota)$ admits only the trivial \virg{Quaternionic} vector bundle \cite[Theorem 1.2 (ii)]{denittis-gomi-14-gen}. As a consequence one gets that ${\rm Ker}(\kappa) =K^{-3}({\n{S}}^1,\iota)$ in this special case. The properties of the homomorphism $\nu$ on $K^{-3}({\n{S}}^1,\iota)$ will be studied in Proposition \ref{prop:iso-one}.
 \hfill $\blacktriangleleft$
\end{remark}

\begin{remark}
As pointed out in  the proof of Theorem \ref{thm:homomorphism_from_KR}, the sewing matrix $q$ depends on the way to express the twisted bundle $X \times \s{H}$ in a standard form. If we change the isomorphism by an automorphism associated to $f : X \to \n{U}(m)$, then the sewing matrix $q$ is changed into a sewing matrix $q'$ and the relation between the two 
sewing matrices is expressed by the equation $q'(x) = f(\tau(x))^\mathtt{t}  q (x)f(x)$.
The difference $\nu_{q'} - \nu_q$ is represented by $[\det [f]]$
in the group $H^1(X,\Z)/H^1_{\Z_2}(X| X^\tau,\Z(1))$. 
Since $\det : [X, \n{U}(m)] \to [X, \n{U}(1)]$ is surjective one has that all the elements in $H^1(X,\Z)/H^1_{\Z_2}(X| X^\tau,\Z(1))$ arise just by changing the isomorphism $X \times \s{H} \cong X \times \C^{2m}$. This means that we cannot define an \emph{absolute} invariant of a band insulator $H$ of class DIII by using the invariant $\nu_q$ of its sewing matrix $q$. Notice, however, that the invariant is able to detect the \emph{difference} of two band insulators, as established in Theorem \ref{thm:homomorphism_from_KR}. In other words  $\nu_q$ works as 
\emph{relative} invariant for band insulators   of class DIII. 
 \hfill $\blacktriangleleft$
\end{remark}

\begin{remark}\label{rk:obos}
A map $\phi:X\to\n{U}(1)$ over the involutive  space $(X,\tau)$ is called \emph{invariant} if $\phi=\phi\circ\tau$.
The equivariant cohomology $H^1_{\Z_2}(X; \Z)$ is isomorphic to the group $[X, \n{U}(1)]_{\mathrm{inv}}$ of homotopy classes of invariant maps (see Appendix \ref{sec:eq:cohom_homot}).  Thus, in general, $[\det[q]] \in H^1_{\Z_2}(X,\Z)$ is an invariant of the homotopy class of a sewing matrix $q : X \to \n{U}(2n)$. By a proof similar to that of Theorem \ref{thm:homomorphism_from_KR}, one can show that this invariant defines a homomorphism
$$
\det \;:\; KR^{-3}(X,\tau)\; \longrightarrow\; H^1_{\Z_2}(X,\Z)\;.
$$
However, this invariant does not detect  $KR^{-3}({\n{S}}^1,\iota)\simeq\Z_2$ since $H^1_{\Z_2}(\n{S}^1,\Z) = 0$. This observation
is in agreement with the content item (II) presented in Section \ref{sect:intro}.
 \hfill $\blacktriangleleft$
\end{remark}

\section{The one-dimensional case}\label{sec:one-dim-case}
In this section we will analyze in detail the case of  
\emph{one-dimensional} band insulators of class DIII. According to Definition \ref{def:int_01}, the latter are band insulators of class DIII defined on the involutive space $({\n{S}}^1,\iota)$. This peculiar case is relevant for its application to  physical systems. In fact, in this section we will show the coincidence of the topological invariant $\nu_q$ described in Definition \ref{dfn:invariant}
 with the topological invariant for topological insulators of class DIII usually used in the physical literature \cite{ryu-schnyder-furusaki-ludwig-10,teo-kane-10,budich-ardonne-13,chiu-teo-schnyder-ryu-16}. For the benefit of the reader we report here a table of some cohomology groups 
of the involutive space  $(\n{S}^1,\iota)$ which will be useful in the following: 
$$
\begin{array}{|c|c|c|c|c|c|}
\hline
n& H^n(\n{S}^1,\Z) & H^n_{\Z_2}(\n{S}^1,\Z) & H^n_{\Z_2}(\n{S}^1,\Z(1)) & H^n_{\Z_2}((\n{S}^1)^\iota,\Z(1))& H^n_{\Z_2}(\n{S}^1|(\n{S}^1)^\iota,\Z(1)) \\
\hline
0 & \Z & \Z & 0 & 0 & 0 \\
\hline
1 & \Z & 0 & \Z \oplus \Z_2 & \Z_2 \oplus \Z_2&\Z \\
\hline
2 & 0 & {\Z_2\oplus\Z_2}  & 0 & 0& 0 \\
\hline
3 & 0 & 0 & {\Z_2\oplus\Z_2} & {\Z_2\oplus\Z_2}& \\
\hline
\end{array}
$$
The proof of the results listed in the table  can be found in \cite[Lemma 2.12]{gomi-15}, in 
\cite[Section 5.3]{denittis-gomi-14} and in \cite[Appendix B]{denittis-gomi-18-I}. Let us also recall that 
$H^1(\n{S}^1,\Z) \simeq [\n{S}^1, \n{U}(1)]$ is generated by the identity map under the identification $ \n{U}(1)\simeq \n{S}^1$. Similarly, in view of  $H^1_{\Z_2}(\n{S}^1,\Z(1)) \simeq[\n{S}^1, \n{U}(1)]_{\Z_2}$ \cite[Proposition A.2]{gomi-15} one gets that the free subgroup $\Z$ is generated by the identity map under the identification of $\n{U}(1)$, endowed with the complex conjugation, with the involutive space $(\n{S}^1,\iota)$. The  torsion subgroup $\Z_2$ is generated by the constant map with  value $-1$. All \virg{Real} and \virg{Quaternionic} vector bundles on $(\n{S}^1,\iota)$ are trivial as proved in \cite[Theorem 1.6]{denittis-gomi-14} and \cite[Theorem 1.2]{denittis-gomi-14-gen}. 
Finally, let us observe that  $(\n{S}^1,\iota)$
satisfies the properties (a), (b) and (c) listed in Assumption \ref{assumption}.

\subsection{The Teo-Kane formula}\label{sec:Teo-Kane}
In this section we will provide an explicit formula to compute the invariant $\nu_q$ in the one-dimensional case.
In particular we will establish the connection with the \emph{Teo-Kane formula}, the latter being the 
topological invariant commonly used in the physical literature 
for one-dimensional for topological insulators of class DIII (see \eg \cite[eq. (4.27)]{teo-kane-10} or \cite[eq. (8)]{budich-ardonne-13} or \cite[eq. (3.70)]{chiu-teo-schnyder-ryu-16}).

\medskip

Let us start with a couple of preliminary results about the  
topology of the maps which define the invariant $\nu_q$.
The first result concerns with the analysis
 of the range of the invariant $\nu_q$.
\begin{lemma}\label{lemma:1D-range}
Consider the involutive space $(\n{S}^1,\iota)$.
Let $\delta:C(\n{S}^1, \n{U}(1))\to \Z_2$ be the map defined by $\delta(p):=(-1)^{{\rm deg}(p)}$ where ${\rm deg}(p)$ denotes the degree (or winding number) of $p$ as a map $p:\n{S}^1\to\n{U}(1)\simeq\n{S}^1$.
Then $\delta$  induces the isomorphism of groups
\begin{equation}\label{eq:target_01}
H^1(\n{S}^1,\Z)/H^1_{\Z_2}\big(\n{S}^1|(\n{S}^1)^\iota,\Z(1)\big) \;\simeq\;\Z_2\;.
\end{equation}
\end{lemma}
\proof
As anticipated in  the introductory part of this section, $H^1_{\Z_2}(\n{S}^1,\Z(1)) \simeq \Z \oplus \Z_2$ is generated by the identity map $\n{S}^1 \to \n{U}(1)$ and the constant map with  value $-1$. 
From the exact sequence \cite[eq. (2.6)]{denittis-gomi-18-I}
$$
\begin{aligned}
0\;=\;H^0_{\Z_2}\big((\n{S}^1)^\iota,\Z(1)\big) &\stackrel{}{\longrightarrow}
H^1_{\Z_2}\big(\n{S}^1|(\n{S}^1)^\iota,\Z(1)\big)\stackrel{\jmath}{\longrightarrow}H^1_{\Z_2}\big(\n{S}^1,\Z(1)\big)\\ 
&\stackrel{r}{\longrightarrow} 
H^1_{\Z_2}\big((\n{S}^1)^\iota,\Z(1)\big) \longrightarrow 
H^2_{\Z_2}\big(\n{S}^1|(\n{S}^1)^\iota,\Z(1)\big)\;=\;0,
\end{aligned}
$$
one gets that the map $r$, induced by the restriction $\n{S}^1 \to (\n{S}^1)^\iota$, is surjective.
Moreover, one infers that the injective image of $H^1_{\Z_2}(\n{S}^1|(\n{S}^1)^\iota,\Z(1))$ under $\jmath$
 is generated by the equivariant map $e:\n{S}^1 \to \n{U}(1)$ given by $e(z):=\expo{ \ii 2k}$. 
Since the map $f:H^1_{\Z_2}(\n{S}^1,\Z(1)) \to H^1(\n{S}^1,\Z)$ 
which forgets the $\Z_2$-action is the projection which induces the bijection on the free part of the two groups, it follows that 
the image of 
$$
H^1_{\Z_2}\big(\n{S}^1|(\n{S}^1)^\iota,\Z(1)\big) \;\stackrel{\jmath}{\longrightarrow}\; H^1_{\Z_2}\big(\n{S}^1,\Z(1)\big) \;\stackrel{f}{\longrightarrow}\;H^1(\n{S}^1,\Z)
$$
is generated by maps $\n{S}^1 \to \n{U}(1)$ with even degree. This proves \eqref{eq:target_01} with the isomorphism  given by the parity of the degree
of the maps from $\n{S}^1$ to $\n{U}(1)$.
\qed

\medskip
 The next result provides a formula for the   degree of the equivariant maps  on the involutive space $(\n{S}^1,\iota)$ with value in $\n{U}(1)$. Let us  emphasize that  a map $r:\n{S}^1\to \n{U}(1)$ is equivariant if $r(\iota(k))=\overline{r(k)}$ for every $k\in \n{S}^1$.
\begin{lemma} \label{lem:degree_formula_on_circle}
Let $r : \n{S}^1 \to \n{U}(1)$ be an equivariant map defined on the involutive space $(\n{S}^1,\iota)$. Then, it holds true that 
$$
(-1)^{\deg (r)} = \frac{r(\pi)}{r(0)}\;.
$$
\end{lemma}
\proof
The equivariance condition  implies that $r(k) = \pm 1$ at the fixed points $k = 0, \pi$. In particular, the value at each fixed point is invariant under equivariant homotopy deformations of $r$. Therefore, the assignment $\Delta: r \mapsto r(\pi)/r(0)$ induces a well-defined homomorphism $\Delta:  H^1_{\Z_2}(\n{S}^1,\Z(1)) \to \Z_2$ (still denoted with the same symbol).
By using the  explicit basis of $H^1_{\Z_2}(\n{S}^1,\Z(1))$, one can verify the equality $\Delta(r) = (-1)^{\deg (r)}$.
\qed

\medskip

With the help of Lemma \ref{lemma:1D-range} and Lemma \ref{lem:degree_formula_on_circle} we can provide a first formula for the invariant $\nu_q$.
\begin{proposition}[The degree fromula]\label{prop:first-from-1D}
Consider the involutive space $(\n{S}^1,\iota)$ along with a map 
 $q : \n{S}^1 \to \n{U}(2n)$ such that
$$
q\big(\iota(k)\big)\;=\;-q(k)^\mathtt{t}\;,\qquad \forall\; k\in \n{S}^1\;.
$$
Then, the topological invariant $\nu_q$ of the map $q$ described in Definition \ref{dfn:invariant} can be computed with the formula
$$
\nu_q\; =\; (-1)^{\deg (p)}\;\in\;\Z_2\;,
$$
independently of  $p \in C(\n{S}^1, \n{U}(1))_q$. 
\end{proposition}
\proof
First of all it is worth noting that all the conditions listed in Assumption \ref{assumption} are satisfied and therefore the invariant $\nu_q$ is well-defined. In view of  Lemma \ref{lemma:1D-range} one gets that 
$$
\nu_q \;\in\;H^1(\n{S}^1,\Z)/H^1_{\Z_2}\big(\n{S}^1|(\n{S}^1)^\iota,\Z(1)\big) \;\simeq\;\Z_2\;.
$$
Since the isomorphism with $\Z_2$ is realized by the parity of the degree, we get the expression $\nu_q = (-1)^{\deg p}$ in terms of $p \in C(\n{S}^1, \n{U}(1))_q$. Indeed, in view of Lemma \ref{lemma_01} (1) we know that for every two $p,p'\in C(\n{S}^1, U(1))_q$ there is an equivariant map $r:\n{S}^1\to \n{U}(1)$ such that $p'=rp$ and  $r(0)=r(\pi)=1$. From Lemma \ref{lem:degree_formula_on_circle} one infers that ${\deg (r)}$ is even and in turn ${\deg (p)}$ and ${\deg (p')}={\deg (p)}+{\deg (r)}$ have the same parity.
\qed

\medskip

The next result is preparatory for the proof of the Teo-Kane formula.

\begin{lemma} \label{lem:homotopy_to_the_trivial_determinant}
Let $q : \n{S}^1 \to \n{U}(2n)$ be a map
defined
 on the involutive space  $(\n{S}^1,\iota)$ 
 such that
$$
q\big(\iota(k)\big)\;=\;-q(k)^\mathtt{t}\;,\qquad \forall\; k\in \n{S}^1\;.
$$
Then, there exists a map
 $\widetilde{q} : \n{S}^1 \times [0, 1] \to \n{U}(2n)$ such that 
 $$
 \widetilde{q}\big(\iota(k),t\big)\; =\; -\widetilde{q}(k, t)^\mathtt{t}\;,\qquad \widetilde{q}(k, 0)\; =\; q(k)\;,\qquad
 {\rm det} \big[\widetilde{q}(k, 1)\big]\;=\;1\;,
 $$
  for all $k \in \n{S}^1$ and $t \in [0, 1]$.
\end{lemma}

\begin{proof}
As anticipated in item ii) of Section \ref{sect:intro},
${\rm det}[q]:\n{S}^1\to\n{U}(1)$ is an invariant map, \ie 
${\rm det}[q(k)]={\rm det}[q(\iota(k))]$ for every $k\in\n{S}^1$.
Let $[\n{S}^1,\n{U}(1)]_{\rm inv}$ be the set of classes of 
homotopy equivalent invariant functions. Then, one has that 
$[\n{S}^1,\n{U}(1)]_{\rm inv}\simeq H^1_{\Z_2}(\n{S}^1,\Z)=0$
(\cf Appendix \ref{sec:eq:cohom_homot}). Therefore,
  there exists ${f} : \n{S}^1 \times [0, 1] \to \n{U}(1)$ such that ${f}(k, t) =  {f}(\iota(k), t)$,
 $ {f}(k, 0) = 1$ and $ {f}(k, 1) = \det [q(k)]^{-1}$ for every $k \in \n{S}^1$ and $t \in [0, 1]$. 
 The degree of the map $k\mapsto  {f}(k, t)$ is constant in $t$ and therefore is trivial for every $t \in [0, 1]$ (since it is trivial for $t=0$). This allows to define  a square root
$$
g(k,t)\;:=\; {f}(k, t)^{\frac{1}{2}}
$$ 
 consistently. This provides a continuous map $g:\n{S}^1 \times [0, 1] \to \n{U}(1)$ such that $g(k, t) = g(\iota(k), t)$,
 $g(k, 0) = 1$ and $g(k, 1)^2 = \det [q(k)]^{-1}$ for every $k \in \n{S}^1$ and $t \in [0, 1]$. Let us introduce the map $  {q}:\n{S}^1 \times [0, 1] \to \n{U}(2n)$ defined by
\begin{equation}\label{eq:q-qtild}
 \widetilde{q}(k, t)\; :=\;\left(\begin{array}{ccc}g(k, t) & 0 & 0 \\0 & 1 & 0 \\0 & 0 & \ddots\end{array}\right) \;q(k)\;\left(\begin{array}{ccc}g(k, t) & 0 & 0 \\0 & 1 & 0 \\0 & 0 & \ddots\end{array}\right)\;.
\end{equation}
 By construction the map $\widetilde{q}$ meets the properties 
stated in the claim. 
\end{proof}

\begin{theorem}[The Teo-Kane formula] \label{thm:invariant_by_Pfaffian}
Let $q : \n{S}^1 \to \n{U}(2n)$ be a map
defined
 on the involutive space  $(\n{S}^1,\iota)$ 
 such that
$$
q\big(\iota(k)\big)\;=\;-q(k)^\mathtt{t}\;,\qquad \forall\; k\in \n{S}^1\;.
$$
Then, the invariant $\nu_q \in \Z_2$ admits the expression
$$
\nu_q \;=\; \frac{{\rm Pf}\big[q(\pi)\big]}{{\rm Pf}\big[q(0)\big]}\; \frac{{\rm det}\big[q(0)\big]^{\frac{1}{2}}}{{\rm det}\big[q(\pi)\big]^{\frac{1}{2}}}\;
$$
where the branch $k\mapsto {\rm det}[q(k)]^{\frac{1}{2}}$ must be chosen continuously between the two
fixed points $k=0,\pi$.\end{theorem}
\proof
Theorem \ref{thm:invariant} (1) ensures that the invariant $\nu_q$ only depends on the homotopy class of $q$. By combining this fact with  Lemma \ref{lem:homotopy_to_the_trivial_determinant} one gets that $\nu_q=\nu_{q'}$  where  
 ${q'}(k):=\widetilde{q}(k,1)$ according to the notation of Lemma \ref{lem:homotopy_to_the_trivial_determinant}.
The crucial property of ${q'}$ is that ${\rm det}[q'(k)]=1$ for every $k\in\n{S}^1$.
From Proposition \ref{prop:first-from-1D} we know that 
$
\nu_{q'} = (-1)^{\deg (p')}
$,
where $p' : \n{S}^1 \to \n{U}(1)$ is such that $1 = p'(\iota(k)) p'(k)$ for all $k \in \n{S}^1$ and $p'(k) = {\rm Pf}[q'(k)]
$ for $k = 0, \pi$. Then, one has $p'(0)^2=1=p'(\pi)^2$
so that $p'(\pi)/p'(0) \in \Z_2$. 
In view of Lemma \ref{lem:degree_formula_on_circle}, one has that 
$$
\nu_{q' }\;=\; (-1)^{\deg(p')}
\;=\;  \frac{p'(\pi)}{p'(0)}\;
=\; \frac{{\rm Pf}\big[q'(\pi)\big]}{{\rm Pf}\big[q'(0)\big]}\;.
$$
From equation \eqref{eq:q-qtild} which provides the link between $q$ and $q'$ and the well-known formula ${\rm Pf}[BAB^\mathtt{t}]={\rm det}[B]{\rm Pf}[A]$ for generic matrices $A,B$ with $A$ skew-symmetric, one gets
$$
{\rm Pf}\big[q'(k)\big]\;=\;g(k,1)\; {\rm Pf}\big[q(k)\big]\;=\;\frac{{\rm Pf}\big[q(k)\big]}{{\rm det}\big[q(k)\big]^{\frac{1}{2}}}\;,\qquad k=0,\pi\;.
$$
This concludes the proof. 
\qed

\medskip

For a concrete physical model of one-dimensional topological insulator of class DIII whose (non-trivial) topology is described
be the Teo-Kane formula we refer to  \cite[Section IV.C]{teo-kane-10} or \cite[Section B.6.a]{chiu-teo-schnyder-ryu-16}.
The connection with the \emph{Kramers polarization} is discussed in 
\cite{budich-ardonne-13}.

\subsection{The role as topological obstruction}
In this section we will show that in dimension one the 
invariant $\nu_q$ can be interpreted as a   \emph{topological obstruction}. 
For this aim, we need to introduce the involutive space $(\n{U}(2n), \theta)$ given by the unitary group on $\C^{2n}$ endowed with the involution 
 $\theta(U): = -U^\mathtt{t}$. The fixed point set $\n{U}(2n)^\theta$ coincides with the set of 
  skew-symmetric matrices on $\C^{2n}$.
Let us introduce the standard symplectic matrix
\begin{equation}\label{eq:mat-Q}
Q \;:=\; 
\left(
\begin{array}{cc}
0 & -{\bf 1}_{n} \\
+{\bf 1}_{n} & 0 
\end{array}
\right)\; \in\; \n{U}(2n)^\theta\;
\end{equation}
and the \emph{compact symplectic group} $Sp(n) = \{ S \in \n{U}(2n)\; |\;  S^\mathtt{t}QS =  Q \}$. One has that ${\rm Pf}[Q]=(-1)^{\frac{n(n+1)}{2}}$. By considering  the left action of 
$Sp(n)$ on $\n{U}(2n)$ one can define the quotient space
$\n{U}(2n)/Sp(n)$. In view of the normal decomposition for skew-symmetric matrices, one  gets the well-defined homeomorphism
\begin{align*}
\n{U}(2n)/Sp(n)\;\ni\;[U]\;\longmapsto\; U^\mathtt{t}QU\;\in\; \n{U}(2n)^\theta\;.
\end{align*}
In particular, $\n{U}(2n)^\theta$ can be seen as the base space of a principal $Sp(n)$-bundle. 

\medskip

Let us start with two preliminary results.
\begin{lemma} \label{lem:normalization_at_base_point}
Let $(X,\tau)$ be an involutive space with a fixed point $\ast \in X^\tau$.
Let $q : X \to \n{U}(2n)$ be a map satisfying 
$$
q\big(\tau(x)\big)\;=\;-q(x)^\mathtt{t}\;,\qquad \forall\; x\in X\;.
$$
Then there exists a map $\widetilde{q} : X \times [0, 1] \to \n{U}(2n)$ such that $\widetilde{q}(\tau(x), t) = -  \widetilde{q}(x, t)^\mathtt{t}$ and $\widetilde{q}(x, 0) = q(x)$ for all $(x,t) \in X\times [0, 1]$ and $\widetilde{q}(\ast, 1) = Q$.
\end{lemma}
\proof
At the fixed point  $q(\ast) \in \n{U}(2n)^\theta$
is a skew-symmetric unitary matrix. Let $U \in \n{U}(2n)$ be a unitary matrix such that $U^\mathtt{t} Q U = q(\ast)$. 
Since $\n{U}(2n)$ is path connected, there is a path ${u} : [0, 1] \to \n{U}(2n)$ such that $u(0) = {\bf 1}_{2n}$ and $u(1) = U^{-1}$. Now, let us define $\widetilde{q} : X \times [0, 1] \to \n{U}(2n)$ by
$$
\tilde{q}(x, t) \;:=\; u(t)^\mathtt{t}\; q(x)\; u(t)\;.
$$
A straightforward check shows that this map has the prescribed properties.
\qed

\begin{lemma}\label{lemma:conn_comp}
The space 
$$
S\n{U}(2n)^\theta\;:=\;\left\{U\in\n{U}(2n)^\theta\;\big|\;{\rm det}[U]=1\right\}
$$
has two connected components distinguished by their Pfaffians. 
\end{lemma}
\proof
Under the homeomorphism $\n{U}(2n)/Sp(n) \simeq \n{U}(2n)^\theta$, one can identify the subspace $S\n{U}(2n)^\theta\subset \n{U}(2n)^\theta$ with
$$
{\det}^{-1}(+1)/Sp(n)\; \sqcup\; {\det}^{-1}(-1)/Sp(n) \;\subset \;\n{U}(2n)/Sp(n)
$$
where ${\det}^{-1}(\pm1):=\{U\in \n{U}(2n)\;|\; {\rm det}(U)=\pm1\}$. In particular ${\det}^{-1}(+1)=S\n{U}(2n)$. Let 
$$
I\;:=\; \left(\begin{array}{ccc}\ii & 0 & 0 \\0 & 1 & 0 \\0 & 0 & \ddots\end{array}\right)\;\in\;\n{U}(2n)\;.
$$
The map $U\mapsto IUI$ provides a homeomorphism ${\det}^{-1}(-1)\simeq{\det}^{-1}(+1)$.
Thus, $S\n{U}(2n)^\theta$ is homeomorphic to the disjoint union of two copies of  $S\n{U}(2n)/Sp(n)$, and in turn it has two connected components. 
From the form of the homeomorphism and the property of the Pfaffian one gets that ${\rm Pf}:[U]\mapsto {\rm Pf}[U^\mathtt{t}QU]=\pm {\rm Pf}[Q]$ for every $[U]\in {\det}^{-1}(\pm1)/Sp(n)$. Therefore,
 the Pfaffian distinguishes the two connected components of $S\n{U}(2n)^\theta$.
 \qed

\medskip

The next result will shed light on the nature of the invariant $\nu_q$
as the topological obstruction for the existence of a homotopy transformation between two given sewing matrices.
\begin{theorem}[Topological obstruction] \label{thm:obstruction_on_circle}
Consider the involutive space $(\n{S}^1,\iota)$ and two maps $q_j:\n{S}^1\to\n{U}(n)$, $j=0,1$, satisfying the sewing matrix condition
$$
q_j\big(\iota(k)\big)\;=\;-q_j(k)^\mathtt{t}\;,\qquad \forall\; k\in \n{S}^1\;,\quad j=0,1\;.
$$
The maps $q_0,q_1$ are homotopy equivalent within the space of sewing matrices if and only if $\nu_{q_0} = \nu_{q_1}$. 
\end{theorem}
\proof
Theorem \ref{thm:invariant} (1) provides the
\virg{only if} part of the claim. Therefore, we 
only need to show the \virg{if} part. By Lemma \ref{lem:homotopy_to_the_trivial_determinant} and Lemma \ref{lem:normalization_at_base_point}, we can assume that $\det [q_0 (k)] = \det [q_1 (k)] = 1$ for all $k \in \n{S}^1$ and  $q_0 (0) = q_1(0) = Q$. Under these assumptions, the coincidence of the invariants $\nu_{q_0} = \nu_{q_1}$ is equivalent to ${\rm Pf}[q_0(\pi)] = {\rm Pf}[q_1(\pi)]
$ in view of  Theorem \ref{thm:invariant_by_Pfaffian}.
We will proceed to build a homotopy 
of sewing matrices $\widetilde{q} : \n{S}^1 \times [0, 1] \to \n{U}(2n)$ such that $\widetilde{q}(k, j) = q_j(k)$ for $j = 0, 1$ and $k \in \n{S}^1$. We will do this in several steps: (1) 
As  a first step let us choose a continuous map 
$$
\widetilde{Q}\; :\; [0, 1]\; \longrightarrow\; S\n{U}(2n)^\theta \subset \n{U}(2n)\;,\qquad \widetilde{Q}(j)\;:=\;q_j(\pi)\;,\;\;j\;=\;0,1\;.
$$
This is possible, because Lemma \ref{lemma:conn_comp} establishes
 that $S\n{U}(2n)^\theta$ has two connected components distinguished by the Pfaffian, and ${\rm Pf}[q_0(\pi)] = {\rm Pf}[q_1(\pi)]$ by hypothesis.\\ (2) As a  second step,  let us  build the continuous map $\tilde{q} : \partial ([0, \pi] \times [0, 1]) \to \n{U}(2n)$ defined by
$$
\widetilde{q}(k, t)
\;:=\;
\left\{
\begin{aligned}
&q_0(0)\; =\; q_1(0) \;=\; Q &\quad\text{if} \quad& k = 0\;,\;\; 0 \leqslant t \leqslant 1\;, \\
&q_0(k)&\quad\text{if} \quad&  0 \leqslant k\leqslant \pi\;,\;\; t = 0\;, \\
&q_1(k)&\quad\text{if} \quad& 0 \leqslant k\leqslant \pi\;,\;\; t = 1\;, \\
&\widetilde{Q}(t)&\quad\text{if} \quad&  k = \pi\;,\;\; 0 \leqslant t \leqslant 1\;.
\end{aligned}
\right.
$$
The map above can be extended, without obstructions, to a continuous map $\widetilde{q} : [0, \pi] \times [0, 1] \to \n{U}(2n)$ (still denoted with the same symbol). Indeed, 
by identifying $\partial([0, \pi] \times [0, 1])$ with $\n{S}^1$, one  can think of $\widetilde{q}$ as a loop in $S\n{U}(2n)$. Since the fundamental group of $S\n{U}(2n)$ is trivial, 
it is possible to get the invoked continuous extension.\\ (3) The third step consists in a further extension
$\widetilde{q} : [-\pi, \pi] \times [0, 1] \to \n{U}(2n)$ defined by
$$
\widetilde{q}(k, t)
\;:=\;
\left\{
\begin{aligned}
&\widetilde{q}(k, t) &\quad\text{if} \quad&  (k, t) \in [0, \pi] \times [0, 1]\;, \\
-&\widetilde{q}(-k, t)^\mathtt{t} &\quad\text{if} \quad&  (k, t) \in [-\pi, 0] \times [0, 1]\;.
\end{aligned}
\right.
$$
Since $q_0(0) = q_1(0) = Q$ is skew-symmetric, $\widetilde{q}$ is a well-defined continuous map on $[-\pi, \pi] \times [0, 1]$. 
Moreover, since $\widetilde{q}(\pi, t) = \widetilde{Q}(t)$ is chosen to be skew-symmetric for all $t \in [0, 1]$, one gets that 
$\widetilde{q}(\pi, t) = \widetilde{q}(-\pi, t)$ for all $t \in [0, 1]$. This shows that $\widetilde{q}$
 descends to a continuous map  $\widetilde{q}:\n{S}^1 \times [0, 1]\to \n{U}(2n)$.\\ (4) The fourth and final step consists in the direct check of the fact that the resulting map $\widetilde{q}$  provides the required homotopy of sewing matrices. 
\qed

\medskip

\begin{remark}[A small generalization]
It is likely that   
Theorem \ref{thm:obstruction_on_circle} could be generalized in the case of a 
 $1$-dimensional $\Z_2$-CW complex $(X,\tau)$ such  that the fixed point set $X^\tau$ is a finite collection of points. In this case it is expected that two sewing matrices $q_0, q_1 : X \to \n{U}(2n)$ are linked by a homotopy within the space of sewing matrices if and only if the following two conditions are satisfied:
\begin{enumerate}
\item
$[\det [q_0]] = [\det [q_1]]$  in $H^1_{\Z_2}(X,\Z) \simeq [X, \n{U}(1)]_{\rm inv}$; 
\vspace{1mm}
\item
$\nu_{q_0} = \nu_{q_1}$ in $H^1(X,Z)/H^1_{\Z_2}(X| X^\tau,\Z(1))$.
\end{enumerate}
Note that   in this case the invariants $\nu_{q_j}$ are well-defined since  
$H^2_{\Z_2}(X| X^\tau,\Z(1)) = 0$ as proved in \cite[Proposition 4.1]{denittis-gomi-18-II}. \hfill $\blacktriangleleft$
\end{remark}

\subsection{The role as an index}
As an application of Theorem \ref{thm:obstruction_on_circle}, we can establish an \emph{index-type} theorem for the invariant for topological insulators of class DIII described in Definition \ref{dfn:invariant}. By this we mean that the quantity $\nu_q$, which is purely \emph{topological} by definition, can be computed by an \emph{analytic} expression. The latter is nothing more than a suitable {index} of an appropriate Fredholm operator.

\medskip

In order to prove the main result of this section we need first to introduce some notation. Let ${\rm Mat}_{2n}(\C)$ be the $C^*$-algebra of $2n\times 2n$ matrices. Every continuous function
$f : \n{S}^1 \to {\rm Mat}_{2n}(\C)$ defines a bounded multiplication operator $M_f$ acting on the Hilbert space 
$L^2(\n{S}^1)\otimes\C^{2n}$ according to the prescription
$$
(M_f\psi)(k)\;:=\;f(k)\psi(k)\;,\qquad \psi\in L^2(\n{S}^1)\otimes\C^{2n}\;.
$$
By construction the map $f\mapsto M_f$ is a homomorphism of $C^*$-algebras. In  particular one has
$M^*_f = M_{f^*}$. Using the standard complex conjugation on $\C^{2n}$, one can endow $L^2(\n{S}^1)\otimes\C^{2n}$ with the real structure (or antiunitary involution) $K$ defined by
 $(K \psi)(k) := \overline{\psi(-k)}$. Given a (bounded) operator $A$ on $L^2(\n{S}^1)\otimes\C^{2n}$, we define the complex conjugate of $A$ as $\overline{A} := KAK$, and the transpose of $A$ as $A^\mathtt{t} := (\overline{A})^* = KA^*K$. The operator $A$ will be called \emph{real} if $\overline{A}=A$, \emph{self-adjoint} if $A=A^*$, \emph{complex symmetric} if $A=A^\mathtt{t}$, \emph{skew-adjoint} if $A=-A^*$ and \emph{skew-complex symmetric} if $A=-A^\mathtt{t}$ \cite{li-zhu-13,ko-ko-lee-15}.

\medskip

The next result shows that 
 the multiplication operators associated to sewing matrices are automatically skew-complex symmetric.

\begin{lemma}\label{lemma:SK-sew-mat}
Let $q : \n{S}^1 \to \n{U}(2n)$ be a map defined on the involutive space  $(\n{S}^1,\iota)$ which satisfies the sewing matrix condition
$$
q\big(\iota(k)\big)\;=\;-q(k)^\mathtt{t}\;,\qquad \forall\; k\in \n{S}^1\;.
$$
Then, it holds true that
$$
(M_q)^\mathtt{t} \:=\: - M_q\;.
$$
\end{lemma}
\proof
One has that
\begin{align*}
(K M^*_q K \psi)(k) &\;=\; \overline{(M_{q^*} K \psi)(-k)}
\;=\; \overline{q(\iota(k))^*}\; \overline{(K \psi)(-k)} \\
&\;=\; - q(k) \psi(k)
= - (M_q \psi)(k)
\end{align*}
for every $\psi\in L^2(\n{S}^1)\otimes\C^{2n}$. This completes the proof.
\qed

\medskip

The Fourier modes $\phi_m(k):=\expo{\ii m k}$, with $m\in\Z$, provides a basis for $L^2(\n{S}^1)$. Let 
$\s{H}_m:=\C[\phi_m]\otimes\C^{2n}$ be the subspace of $L^2(\n{S}^1)\otimes\C^{2n}$ generated by $\phi_m\otimes {\rm v}$ with ${\rm v}\in \C^{2n}$ and define the subspace of positive frequencies
$\s{H}_+:=\bigoplus_{m\in\N}\s{H}_m$. The associated 
orthogonal projection $\Pi:L^2(\n{S}^1)\otimes\C^{2n}\to \s{H}_+$
is known as the \emph{Hardy projection}.
Its adjoint is given by the inclusion $\Pi^* : \s{H}_+\to L^2(\n{S}^1)\otimes\C^{2n}$. Since $K\phi_m=\phi_m$ for every $m\in\Z$, it follows that  $\s{H}_+$ is preserved by the real structure $K$. 
Therefore (with a little abuse of notation) one has that 
$K \Pi = \Pi K$ and $K \Pi^* = \Pi^* K$. 
The \emph{Toeplitz operator} associated to the continuous function $f : \n{S}^1 \to {\rm Mat}_{2n}(\C)$ is the bounded operator $T_f$ on $\s{H}_+$
defined by the compression $T_f := \Pi M_f \Pi^*$. As a matter of fact, if $f$ vanishes nowhere, then  $T_f$ is a Fredholm operator \cite[Corollary 3.5.12]{murphy-13} with the associated {Noether-Fredholm index} \cite[Section 1.4]{murphy-13} 
$$
{\rm ind}(T_f)\;:=\;{\rm dim}\; {\rm Ker} (T_f)\;-\;{\rm dim}\; {\rm Ker} (T_f^*)\;\in\;\Z\;.
$$
However, when  $T_f$ is skew-complex symmetric
one automatically gets ${\rm ind}(T_f)=0$. In fact from 
$(T_f)^\mathtt{t}= - T_f$ one
immediately infers the equality of the dimension of ${\rm Ker} (T_f)$ and ${\rm Ker} (T_f^*)$.
 Therefore, to capture the topology of  skew-complex symmetric Toeplitz operators one needs a secondary index. Following \cite{atiyah-singer-69,denittis-schulz-baldes-14,schulz-baldes-15}, we will consider the \emph{$\Z_2$-(Fredholm) index} defined by 
\begin{equation}\label{eq:Z2-ind}
{\rm ind}_{\Z_2}(T_f)\;:=\;(-1)^{{\rm dim}\; {\rm Ker} (T_f)}\;\in\;\Z_2\;.
\end{equation}
This index turns out to be  a deformation invariant for a skew-complex symmetric Fredholm operator. To justify the latter claim, let us consider the unitary operator on $L^2(\n{S}^1)\otimes\C^{2n}$ defined by
$$
\s{I}\;:=\;{\bf 1}_{L^2}\;\otimes\; \left(
\begin{array}{cc}
 0& -{\bf 1}_{n} \\
  +{\bf 1}_{n} & 0
\end{array}
\right)\;.
$$
One can check that $\s{I}^{-1}=\s{I}^*=-\s{I}$, and $K\s{I}K=\s{I}$, namely $\s{I}$ is \emph{real skew-symmetric}. Moreover,  
$\s{I}$ commutes with   the Hardy projection $\Pi$, and therefore it provides a real skew-symmetric operator also on the subspace $\s{H}_+$. If $A$ is skew-complex symmetric Fredholm operator then 
$A':=\s{I}A$ satisfies the relation $\s{I}^*(A')^\mathtt{t}\s{I}=A'$, namely $A'$ is an \emph{odd symmetric operator} in the parlance of \cite{denittis-schulz-baldes-14,schulz-baldes-15}.
In view of the bijection induced by the mapping $A\mapsto \s{I}A$, and the fact that the dimension of the kernels of $A$ and $\s{I}A$ are the same, we can  use \cite[Theorem 2]{schulz-baldes-15} to state:
\begin{lemma}\label{lemma_top_ind}
The set of skew-complex symmetric Toeplitz operators on the Hilbert space $\s{H}_+$ is  the disjoint union of two open and connected components labelled by ${\rm ind}_{\Z_2}$.
\end{lemma}

\medskip

 Therefore, as a consequence of Lemma \ref{lemma:SK-sew-mat} and Lemma \ref{lemma_top_ind}, it turns out that the $\Z_2$-index is the correct invariant to detect the topology of the Toeplitz operator $T_q=\Pi M_q \Pi^*$ associated to a sewing matrix $q : \n{S}^1 \to \n{U}(2n)$.
Since a homotopy of sewing matrices  induces a homotopy of the associated Toeplitz operators in the space of
skew-complex symmetric Fredholm operators it follows that ${\rm ind}_{\Z_2}(T_q)$
provides  a homotopy invariant of the sewing matrix $q$. The next result shows that this invariant coincides with $\nu_q$.
\begin{theorem}[Index theorem]\label{the:index}
Let $q : \n{S}^1 \to \n{U}(2n)$ be a map defined on the involutive space  $(\n{S}^1,\iota)$ which satisfies the sewing matrix condition
$$
q\big(\iota(k)\big)\;=\;-q(k)^\mathtt{t}\;,\qquad \forall\; k\in \n{S}^1\;.
$$
Then, it holds true that
$$
\nu_q \;=\;  {\rm ind}_{\Z_2}(T_q)\;.
$$
In particular, this implies that 
$$
\frac{{\rm Pf}\big[q(\pi)\big]}{{\rm Pf}\big[q(0)\big]}\; \frac{{\rm det}\big[q(0)\big]^{\frac{1}{2}}}{{\rm det}\big[q(\pi)\big]^{\frac{1}{2}}}
 \;=\; (-1)^{{\rm dim}\; {\rm Ker} (T_q)}\;.
$$
where the branch $k\mapsto {\rm det}[q(k)]^{\frac{1}{2}}$ must be chosen continuously between the two
fixed points $k=0,\pi$.

\end{theorem}
\proof
It is enough to check the coincidence of the invariants for a representative of each of the 
\emph{two}
homotopy classes of sewing matrices (\cf Theorem \ref{thm:obstruction_on_circle}). 
One homotopy class is represented by the constant map
$q_+(k) = Q$ for all $k\in\n{S}^1$, with $Q$ defined by \eqref{eq:mat-Q}. An immediate application of Theorem \ref{thm:invariant_by_Pfaffian} provides $\nu_{q_+}=+1$.
 On the other hand, since $M_{q_0}$ is the multiplication operator by a constant invertible matrix, it follows that $T_{q_+}:=\Pi M_{q_+} \Pi^*$ is invertible in $\s{H}_+$. As a consequence 
the dimension of ${\rm Ker} (T_{q_+})$ is zero, and in turn 
${\rm ind}_{\Z_2}(T_{q_+})=+1$. This shows that $\nu_{q_+}={\rm ind}_{\Z_2}(T_{q_+})$. The second homotopy class can be represented by
$$
q_-(k)\;: =\;
\left(
\begin{array}{cc|cc}
0 & +\expo{+\ii k} & 0 & 0 \\
-\expo{-\ii k} & 0 & 0 & 0 \\
\hline
0 & 0 & 0& -{\bf 1}_{n-1} \\
0 & 0 & +{\bf 1}_{n-1} & 0
\end{array}
\right)\;,\qquad k\in\n{S}^1\;.
$$
Indeed, from
$$
\begin{aligned}
{\rm Pf}\big[q_-(k)\big]\;&=\;{\rm Pf}\left[
\left(\begin{array}{cc}
0 & +\expo{+\ii k} \\
-\expo{-\ii k} & 0  \\
\end{array}
\right)\right]\;{\rm Pf}\left[
\left(
\begin{array}{cc}
 0& -{\bf 1}_{n-1} \\
  +{\bf 1}_{n-1} & 0
\end{array}
\right)\right]\\
&=\;(-1)^{\frac{n(n-1)}{2}}\;{\rm Pf}\left[
\left(\begin{array}{cc}
0 & +\expo{+\ii k} \\
-\expo{-\ii k} & 0  \\
\end{array}
\right)\right]
\end{aligned}\;,
\qquad k=0,\pi
$$
one gets that 
$$
{\rm Pf}\big[q_-(0)\big]\;=\;(-1)^{\frac{n(n-1)}{2}}\;=\;-{\rm Pf}\big[q_-(\pi)\big]\;,
$$
and in turn
one  has $\nu_{q_-} = -1$ as a consequence of Theorem \ref{thm:invariant_by_Pfaffian} and  the constancy of the determinant
\[
{\rm det}\big[q_-(k)\big]\;=\;{\rm det}\left[\left(
\begin{array}{cc}
 0& -{\bf 1}_{n-1} \\
  +{\bf 1}_{n-1} & 0
\end{array}
\right)\right]\;=\;(-1)^{n-1}\;,\qquad \forall k\in\n{S}^1\:.
\]
To compute the kernel of $T_{q_-}$ let us consider the orthogonal decomposition  
$\s{H}_+=\s{H}_{+,1}\oplus\s{H}_{+,2}$ induced by the decomposition
$$
\C^{2n}\;=\;\C[e_1,e_2]\;\oplus\;\C[e_3,\ldots,e_{2n}]\;\simeq\;\C^{2}\oplus \C^{2(n-1)}
$$ where $\{e_1,\ldots,e_{2n}\}$ denotes the canonical basis of $\C^{2n}$. Evidently 
$T_{q_-}$ is injective on $\s{H}_{+,2}$. On the other hand, a generic element of $\psi\in \s{H}_{+,1}$ can be represented in terms of the Fourier modes $\phi_m$ as
$$
\psi\;=\;\left(\sum_{m\in\N}a_{m,1}\;\phi_m\otimes e_1\right)+\left(\sum_{m\in\N}a_{m,2}\;\phi_m\otimes e_2\right)\;,\qquad a_{m,\ell}\in\C\;,\quad \ell=1,2
$$
and a direct computation provides
$$
M_{q_-} \Pi^*\psi\;=\;\left(\sum_{m\in\N}-a_{m,1}\;\phi_{m-1}\otimes e_2\right)+\left(\sum_{m\in\N}a_{m,2}\;\phi_{m+1}\otimes e_1\right)\;.
$$
Therefore, one gets that the equation $T_{q_-} \psi$ has the non-trivial solution $\psi=a(\phi_1\otimes e_1)$, $a \in\C$, and as a consequence the dimension of ${\rm Ker} (T_{q_-})$ is one.
This provides ${\rm ind}_{\Z_2}(T_{q_-1})=-1$ and in turn  $\nu_{q_-}={\rm ind}_{\Z_2}(T_{q_-})$. This concludes the proof.
\qed

\begin{remark}[The bulk-edge correspondence]
In the jargon of the theory of topological insulator 
an index formula of the type of Theorem \ref{the:index} is referred as \emph{bulk-edge correspondence} \cite{prodan-schulz-baldes-book}. In fact the topological invariant $\nu_q$ is derived by the data concerning the system in the \emph{bulk}, while the index ${\rm ind}_{\Z_2}(T_q)$ is computed  by means of a Toeplitz operator which encodes the existence of an edge. 
\hfill $\blacktriangleleft$
\end{remark}

\begin{remark}[The interpretation in K-theory]
By combining  the map $A\mapsto \s{I}A$ introduced above with  the argument used in the proof of \cite[Theorem 2 ]{schulz-baldes-15} one obtain a homeomorphism between the 
set of skew-complex symmetric Toeplitz operators on the Hilbert space $\s{H}_+$ and the classifying space $\s{F}^2(\s{H}_\R)$ introduced in \cite{atiyah-singer-69}. As a consequence
the $K$-group $KR^{-2}(\ast)$ can be identified with the set of homotopy class of skew-complex symmetric Fredholm operators, and the $\Z_2$-index \eqref{eq:Z2-ind} realizes the isomorphism $KR^{-2}(\ast) \simeq \Z_2$. The construction of the skew-complex symmetric Fredholm operator $T_q$ from $q$ extends to a homomorphism of $KR$-theory
$$
T \;:\; KR^{-3}(\n{S}^1,\iota)\; \longrightarrow\; KR^{-2}(\ast)\;,
$$
which is bijective by the results obtained so far. Notice that the push-forward
$$
\pi_* \;:\; KR^{-3} (\n{S}^1,\iota)\; \longrightarrow\; KR^{-2}(\ast)
$$
along the projection $\pi : (\n{S}^1,\iota) \to \ast$ is also a bijection. Thus, we have $\pi_* = T$ which means that  the construction with the Toeplitz operator realizes the push-forward.
 \hfill $\blacktriangleleft$
\end{remark}

\subsection{The relation with the KR-theory}\label{sec:KR-theo_II}
In this section we will prove that in the one-dimensional case the 
homomorphism \eqref{int_nu} is indeed an isomorphism.
In Lemma \ref{lemma:1D-range} it  has been proved that 
$$
H^1(\n{S}^1,\Z)/H^1_{\Z_2}\big(\n{S}^1|(\n{S}^1)^\iota,\Z(1)\big) \;\simeq\;\Z_2\;.
$$
On the other hand, by using formula (B.4) and  Table B.1 in \cite{denittis-gomi-14}  one gets
$$
\begin{aligned}
KR^{-3} (\n{S}^1,\iota)\;&\simeq\;KR^{-3}(\ast)\;\oplus\;KR^{-2}(\ast)\\
&\simeq\;0\;\oplus\;\Z_2\;\simeq\;\Z_2\;.
\end{aligned}
$$
Theorem \ref{thm:homomorphism_from_KR} and Remark \ref{rk:obos-S1} establish  the existence of  the well-defined homomorphism
\begin{equation}\label{eq:one-dim_sio-nu}
\nu \;:\; KR^{-3}(\n{S}^1,\iota)\; \longrightarrow\; H^1(\n{S}^1,\Z)/H^1_{\Z_2}\big(\n{S}^1|(\n{S}^1)^\iota,\Z(1)\big)\;.
\end{equation}
\begin{proposition}\label{prop:iso-one}
The homomorphism \eqref{eq:one-dim_sio-nu} is  an isomorphism.
\end{proposition}
\proof
From the computations above one has that  the homomorphism \eqref{eq:one-dim_sio-nu} has the form $\nu:\Z_2\to\Z_2$. 
Moreover, in the proof of Theorem \ref{the:index} we construct two
sewing matrices $q_\pm$ such that $\nu_{q_\pm}=\pm1$. Let $[q_\pm]\in KR^{-3}(\n{S}^1,\iota)$ be the class of the 
 band insulators represented by the sewing matrix $q_\pm$. Since the homomorphism \eqref{eq:one-dim_sio-nu} corresponds to $\nu:[q_\pm]=\nu_{q_\pm}$ one gets the result.
 \qed

\subsection{The relation with the classification of  \virg{Real} gerbes}\label{sec:gerbs}
The aim of this section is to clarify  the relation between the  invariant $\nu$ and a \emph{gerbe invariant} introduced in \cite{gomi-thiang-21}. 

\medskip

Let us first introduce the  
 gerbe invariant in a way slightly different from the original definition in \cite{gomi-thiang-21}. First of all, let us recall that  a gerbe $\bb{G}$ over the space $S\n{U}(2n)$ is classified by its 
 Dixmier-Douady class 
 $$
 \s{DD}[\bb{G}]\;\in\;H^3\big(S\n{U}(2n), \Z\big)\; \simeq\; \Z\;.
 $$
The basic gerbe $\bb{G}$ over $S\n{U}(2n)$ is a representative of the class $ \s{DD}[\bb{G}]\simeq 1$ which generates $H^3(S\n{U}(2n), \Z)$.
The basic  gerbe $\bb{G}$ admits only two \virg{Real} structures   (up to stable equivalences) over the involutive space $(S\n{U}(2n),\theta)$. We will denote with $\bb{G}_+$ and  
$\bb{G}_-$ the 
 basic gerbe equipped with the two different  `Real'' structures.
 The two \virg{Real} basic  gerbes are classified by their  \virg{Real} Dixmier-Douady class
$$
 \s{DD}_R[\bb{G}_\pm]\;\in\; H^3_{\Z_2}\big(S\n{U}(2n), \Z(1)\big)\;
$$ 
 and the two classes differ by an element of in $H^3_{\Z_2}(\ast, \Z(1)) \subset H^3_{\Z_2}(S\n{U}(2n), \Z(1))$, where $\ast \in S\n{U}(2n)$ is any fixed point under the involution $\theta$.
 As a result, we get   a well-defined element
 $$
 \widehat{\s{DD}}_R[\bb{G}_\pm]\; \in\; H^3_{\Z_2}\big(S\n{U}(2n), \Z(1)\big)/H^3_{\Z_2}\big(\ast, \Z(1)\big)\;,
 $$  
 independently  of the choice of the \virg{Real} structures on $\mathcal{G}$, as well as of the fixed point $\ast \in S\n{U}(2n)$.
 
 \medskip
 
 Let $(X, \tau)$ be an involutive space which admits a fixed point $\ast \in X$, and $q : X \to S\n{U}(2n)$ an equivariant map. 
 Such a map induces by pull-back the homomorphism
 $$
 q^*\;:\;H^3_{\Z_2}\big(S\n{U}(2n), \Z(1)\big)/H^3_{\Z_2}\big(\ast, \Z(1)\big)\;\longrightarrow\;H^3_{\Z_2}\big(X, \Z(1)\big)/H^3_{\Z_2}\big(\ast, \Z(1)\big)\;.
 $$
 The \emph{gerbe invariant} of the map $q : X \to S\n{U}(2n)$ is by definition the pull-back
 $$
dd_R[q]\;:=\;q^*\widehat{\s{DD}}_R[\bb{G}_\pm]\;.
 $$
  This is an invariant of the equivariant homotopy class of $q$. Therefore, $\dd{d}_R$ provides a well-defined  map (still denoted with the same symbol)
\begin{align*}
dd_R \;:\; \big[X, S\n{U}(2n)\big]_{\Z_2}\; \longrightarrow\; 
H^3_{\Z_2}\big(X, \Z(1)\big)/H^3_{\Z_2}\big(\ast, \Z(1)\big).
\end{align*}
The invariants $dd_R[q]$ and $\nu_q$ take values in different groups, and therefore we cannot generally compare these two invariants. However, in the case of the involutive space $(\n{S}^1,\iota)$, one has that  
$$
H^3_{\Z_2}\big(\n{S}^1, \Z(1)\big)/H^3_{\Z_2}\big(\ast, \Z(1)\big)\; \simeq\; \Z_2
$$
 in view of 
 $$
 H^3_{\Z_2}\big(\n{S}^1, \Z(1)\big)\; \simeq\; \tilde{H}^3_{\Z_2}\big(\n{S}^1, \Z(1)\big) \oplus H^3_{\Z_2}\big(\ast, \Z(1)\big)\; \simeq\; \Z_2 \oplus \Z_2
 $$ and 
 $H^3_{\Z_2}(\ast, \Z(1)) \simeq \Z_2$. This allows us to regard $dd_R[q] \in \Z_2$. Since also $\nu_q \in \Z_2$
 we can compare the two $\Z_2$-invariants $dd_R[q]$ and $\nu_q$ for equivariant maps $q : \n{S}^1 \to S\n{U}(2n)$ over the involutive circle $(\n{S}^1,\iota)$.
\begin{theorem}[Gerbe invariant]
Let $q : \n{S}^1 \to S\n{U}(2n)$ be an equivariant map defined on
 the involutive space  $(\n{S}^1,\iota)$ 
 satisfying the sewing matrix condition
$$
q\big(\iota(k)\big)\;=\;-q(k)^\mathtt{t}\;,\qquad \forall\; k\in \n{S}^1\;.
$$
Then, it holds true that
$$
\nu_q \;=\;  dd_R[q]\;.
$$
\end{theorem}
\proof
The proof is similar to that of Theorem \ref{the:index}. By Lemma \ref{lem:normalization_at_base_point} and Theorem \ref{thm:obstruction_on_circle}, the invariant $\nu$ induces a bijection $[\n{S}^1, S\n{U}(2n)]_{\Z_2} \simeq \Z_2$. In the proof of Theorem \ref{the:index}, we constructed two equivariant maps $q_\pm : \n{S}^1 \to S\n{U}(2n)$ such that $\nu_{q_\pm} = \pm 1$.
Therefore, the claim is proved by showing that $dd_R[q_\pm]=\pm 1$.
For the computation, we make use of an expression of 
$dd_R[q_\pm] \in \Z_2$  in terms of the \emph{sign invariant} of \virg{Real} gerbes \cite[Definition 4.4]{gomi-thiang-21}.
 In general, for a \virg{Real} gerbe $\bb{G}$ on an involutive space $(X, \tau)$ and a path connected subspace $Y = Y^\tau \subset X^\tau$ in the fixed point set, there is a well-defined sign $\sigma(\bb{G}, Y) \in \Z_2$. 
 Actually, this is the component in $H^0(Y, \Z_2) \simeq \Z_2$ of the ``Real'' Dixmier-Douady class $\s{DD}_R(\bb{G}|_Y) \in H^3_{\Z_2}(Y, \Z(1))$ of the restricted gerbe $\bb{G}|_Y$ with respect to the decomposition 
 $$
 H^3_{\Z_2}\big(Y, \Z(1)\big)\; \simeq\; H^0(Y, \Z_2)\; \oplus\; H^2(Y, \Z_2)
 $$ 
 as shown in \cite[Lemma B.1]{gomi-thiang-21}. 
 The fixed point set of the involutive space $(\n{S}^1,\iota)$ consists of the points  $k_+ = 0$ and $k_- = \pi$. 
Let $\bb{G}_\pm$ be the two \virg{Real} basic gerbes  over $S\n{U}(2n)$ and $q^*\bb{G}_\pm$ the pull-back gerbes over $\n{S}^1$.
 In view of \cite[Corolary 4.7]{gomi-thiang-21} one has that
 $$
 \begin{aligned}
 q^*\s{DD}_R(\bb{G}_\pm)\;&=\;\s{DD}_R(q^*\bb{G}_\pm)\;=\;\sigma\big(q^*\bb{G}_\pm, \{k_+\}\big)\;\sigma\big(q^*\bb{G}_\pm, \{k_-\}\big)\\
&=\; \sigma\big(\bb{G}_\pm, \{q(k_+)\}\big)\;\sigma\big(\bb{G}_\pm, \{q(k_-)\}\big)
 \end{aligned}
 $$
 is computed as a product of signs. Since the right-hand side of the equation above is not affected to the passage to the quotient which defines the map $dd_R$, one gets
 $$
 dd_R[q]\;=\;\sigma\big(\bb{G}_\sharp, \{q(k_+)\}\big)\;\sigma\big(\bb{G}_\sharp, \{q(k_-)\}\big)\;\in\;\Z_2
 $$
 where $\bb{G}_\sharp$ stands irrelevantly for $\bb{G}_+$ or $\bb{G}_-$.
 As seen in Lemma \ref{lemma:conn_comp}, the fixed point set of  $(S\n{U}(2n), \theta)$ has two connected components. As it is pointed out in \cite[Remark 5.7]{gomi-thiang-21}, the sign invariant of  $\bb{G}_\sharp$ is constant on the connected components of $S\n{U}(2n)^\theta$ and differs on each of the two components.
 As a result, $dd_R[q] = + 1$ if and only if $q(k_+)$ and $q(k_-)$ are in the same  connected component
and  
 $dd_R[q] = - 1$
 if and only if $q(k_+)$ and $q(k_-)$ are in different connected components. This  immediately proves that  $dd_R[q_\pm] = \pm 1$ and the proof is over.
\qed

\section{The two-dimensional case}\label{sec:two-dim-case}
The analysis of physical systems 
shows that there are (at least) two interesting
 definitions for two-dimensional band insulators of class DIII
which differ for the involutive base space. By borrowing the discussion in \cite[Section 2]{denittis-gomi-14}, we can distinguish between the \emph{Dirac} (or \emph{free}) case and the \emph{Bloch} (or \emph{periodic}) case.
In the Dirac case  the involutive  space
is  $(\n{S}^2,\iota)$ where
$$
\n{S}^2\;:=\;\left\{(x_0,x_1,x_2)\in\R^3\;\big|\;x_0^2+x_1^2+x_2^2=1\right\}
$$
is the two-dimensional sphere with involution given by
$$
\iota(x_0,x_1,x_2)\;:=\; (x_0,-x_1,-x_2)\;,\qquad \forall\;(x_0,x_1,x_2)\in\n{S}^2\;.
$$
In the Bloch case  the involutive  space
is  $(\n{T}^2,\iota)$ where
$$
\n{T}^2\;:=\;\n{S}^1\times\n{S}^1\;=\;\left(\R/2\pi\Z\right)\times \;\left(\R/2\pi\Z\right)
$$
is the two-dimensional torus with involution given by
$$
\iota(k_1,k_2)\;:=\; (-k_1,-k_2)\;,\qquad \forall\;(k_1,k_2)\in\n{T}^2\;.
$$
In short, one has that $(\n{T}^2,\iota)=(\n{S}^1,\iota)\times (\n{S}^1,\iota)$ is given by two copies of the involutive  space
for the one-dimensional case. 

\medskip

According to the terminology currently used in the theory of topological insulators, the topological phases detected in the Dirac case are called \emph{strong}. Usually, the Bloch case present a richer family of topological phases.
Some of these are inherited from the Dirac case via the pull-back
of the equivariant map $\pi_0:(\n{T}^2,\iota)\to (\n{S}^2,\iota)$
described in \cite[eq. (4.14)]{denittis-gomi-14-gen}
and are still called strong. The remaining topological phases
are known as \emph{weak}.

\medskip

It is worth noting that
both spaces  $(\n{S}^2,\iota)$ and $(\n{T}^2,\iota)$ meet the properties (a) and (b) listed in Assumption \ref{assumption}. However, the  isomorphisms
$$
H^2_{\Z_2}\big(\n{S}^2|(\n{S}^2)^\iota,\Z(1)\big)\;\stackrel{}{\simeq}\;\Z_2\;\stackrel{}{\simeq}\;H^2_{\Z_2}\big(\n{T}^2|(\n{T}^2)^\iota,\Z(1)\big)\,,
$$
and the bijectivity of the FKMM invariant
 \cite[Theorem 1.2]{denittis-gomi-14-gen}, imply that
the property (c) of Assumption \ref{assumption}
 is not automatically guaranteed. This fact  makes the analysis of the two-dimensional case slightly more complicated than the one-dimensional case.

\subsection{Strong invariant: the Dirac case} 
Let us start by observing that for a two-dimensional sphere it holds true that $H^1(\n{S}^2,\Z)=0$ which immediately implies that
$$
H^1(\n{S}^2,\Z)/H^1_{\Z_2}(\n{S}^2|(\n{S}^2)^\iota,\Z(1))\;=\;0\;.
$$
Therefore, in the Dirac case the map $\nu$ in \eqref{int_nu} amounts to the trivial homomorphism, and for this reason it is unable to capture the topology of the related band insulator of class DIII.

\medskip

On the other hand,  from equations (B.3) and (B.9) and Tables B.1 and B.3 in \cite{denittis-gomi-14}  one gets
$$
\begin{aligned}
KR^{-3} (\n{S}^2,\iota)\;&\simeq\;\widetilde{KR}^{-3}(\n{S}^2,\iota)\;\oplus\;KR^{-3}(\ast)\\
&\simeq\;\widetilde{KR}^{0}(\n{S}^6,\iota)\;\oplus\;0\;\simeq\;\Z_2\;.
\end{aligned}
$$
The latter equation shows that there is a non-trivial topology for 
band insulators of class DIII over $(\n{S}^2,\iota)$.

\medskip

The topology of $KR^{-3} (\n{S}^2,\iota)$ can be detected via the FKMM invariant by using the map $\kappa$ in \eqref{eq:FKMM-KR}.
First of all let us observe  that 
$$
 {\rm Vec}^{2n}_{\rr{Q},0}({\n{S}}^2,\iota)\;\simeq\;{\rm Vec}^{2}_{\rr{Q}}({\n{S}}^2,\iota)\;\simeq\;H^2_{\Z_2}\big(\n{S}^2|(\n{S}^2)^\iota,\Z(1)\big)\;\simeq\;\Z_2\;.
$$
The first isomorphism  is a consequence of  \cite[Proposition 4.1]{denittis-gomi-14-gen}
which provides the equality ${\rm Vec}^{2n}_{\rr{Q},0}({\n{S}}^2,\iota)={\rm Vec}^{2n}_{\rr{Q}}({\n{S}}^2,\iota)$
(\ie every   \virg{Quaternionic} vector bundle over 
 $(\n{S}^2,\iota)$ can be built over a product vector bundle)  and \cite[Corollary 2.1]{denittis-gomi-14-gen} which describes the stable range for low-dimensional \virg{Quaternionic} vector bundles. The second isomorphism is provided by the FKMM invariant \cite[Theorem 1.2]{denittis-gomi-14-gen}. With this information we can prove the following result:
\begin{proposition}[Two-dimensional strong invariant]\label{prop_strong}
The homomorphism
$$
\kappa\;:\;KR^{-3} (\n{S}^2,\iota)\;\longrightarrow\;H^2_{\Z_2}\big(\n{S}^2|(\n{S}^2)^\iota,\Z(1)\big)
$$
described in Theorem \ref{thm:homomorphism_from_KR} is indeed an isomorphism.
\end{proposition}
\proof
In view of Theorem \ref{thm:homomorphism_from_KR} and the results showed above we are considering a well-defined homomorphism of the form $\kappa:\Z_2\to\Z_2$. Therefore, to prove the claim  it is enough to show that $\kappa$ is surjective. Let us use the additive representation of $\Z_2\simeq\Z/2\Z$.
The two distinct classes in  ${\rm Vec}^{2}_{\rr{Q}}({\n{S}}^2,\iota)\simeq\Z_2$ can be represented by the \virg{Quaternionic} vector bundles $\bb{E}_{q_0}$ and $\bb{E}_{q_1}$ obtained by
endowing the product bundles $\n{S}^2\times\C^2$ with two inequivalent \virg{Quaternionic} structures $\Theta_{q_0}$ and $\Theta_{q_1}$ associated to two sewing matrices $q_0,q_1:{\n{S}}^2\to\n{U}(2)$. 
Let $q_0$ be the constant map defined as in Theorem \ref{thm:invariant} (2). Then,  $\bb{E}_{q_0}$ turns out to be the trivial \virg{Quaternionic} vector bundle with FKMM invariant $\kappa(\bb{E}_{q_0})=0$. An explicit realization for the sewing matrix $q_1$ is described in \cite[eq. (4.9)]{denittis-gomi-14-gen}. In this case $\bb{E}_{q_1}$ provides a representative for the  non trivial \virg{Quaternionic} class and $\kappa(\bb{E}_{q_1})=1$.
Now, consider the 
 Hamiltonians $H_i:{\n{S}^2}\to \rr{Herm}(\C^{4})^\times$ given by 
$$
H_i(x)\;:=\;
\left(\begin{array}{cc}0 & q_i(x)^{-1}\\
q_i(x) & 0\end{array}\right)\;,\qquad x\in \n{S}^2\;,\;\;i=0,1\;.
$$
Then $[\n{S}^2\times\C^4, H_1,H_0]$ is a well-defined class in $KR^{-3} (\n{S}^2,\iota)$ and by equation \eqref{eq:def_k_KR} one gets $\kappa([\n{S}^2\times\C^4, H_1,H_0])=1$. This concludes the proof.
\qed

\medskip

Let $H:{\n{S}^2}\to \rr{Herm}(\C^{4n})^\times$ be a band insulator of class DIII over $(\n{S}^2,\iota)$ represented in a standard way. Denote with $H_0$ the trivial insulator constructed with the constant sewing matrix $q_0:\n{S}^2\to\n{U}(2n)$ described 
in item (2) of Theorem \ref{thm:invariant}. Then
$$
\nu^{\rm s}(H)\;:=\;\kappa\big([\n{S}^2\times\C^{4n}, H,H_0]\big)
$$
will be called the \emph{strong invariant} of 
$H$. This is nothing more than the FKMM invariant $\kappa(\bb{E}_{q_H})$ of the \virg{Quaternionic} vector bundle associated to $H$ with the sewing matrix $q_H$.
 In view of Proposition \ref{prop_strong} this invariant totally captures the topology of   two-dimensional band insulators  of class DIII in the Dirac case.
 
\subsection{Weak-strong invariant: the Bloch case} 
In order to study the classification of two-dimensional band insulator $H$ of class DIII in the Bloch case we will specialize  
Theorem \ref{thm:homomorphism_from_KR} for the involutive space $(\n{T}^2,\iota)$.

\medskip

By using \cite[eq. (B.7)]{denittis-gomi-14} one gets
\begin{equation}\label{eqKR3-T2}
\begin{aligned}
KR^{-3} (\n{T}^2,\iota)\;&\simeq\;{KR}^{-3}(\ast)\;\oplus\;\left(KR^{-2}(\ast)\right)^{\oplus 2}\;\oplus\;KR^{-1}(\ast)\\
&\simeq\;0\;\oplus\;\left(\Z_2\right)^{\oplus 2}\;\oplus\;\Z_2\\&\simeq\;\Z_2\oplus\;\Z_2\oplus\;\Z_2\;.
\end{aligned}
\end{equation}
On the other hand,  the product structure $(\n{T}^2,\iota):=(\n{S}^1,\iota)\times (\n{S}^1,\iota)$ 
and the use of the K\"{u}nneth formula, 
provide
\begin{equation}\label{eq:comput_cohom_T2}
\begin{aligned}
H^1(\n{T}^2,\Z)/H^1_{\Z_2}\big(\n{T}^2|(\n{T}^2)^\iota,\Z(1)\big)\;&\simeq\;\Big(H^1(\n{S}^1,\Z)/H^1_{\Z_2}\big(\n{S}^1|(\n{S}^1)^\iota,\Z(1)\big)\Big)^{\oplus 2}\\
 &\simeq\;\left(\Z_2\right)^{\oplus 2}\;=\;\Z_2\oplus\Z_2\;.
\end{aligned}
\end{equation}
Finally, from \cite[Theorem 1.2]{denittis-gomi-14-gen} one knows that 
$$
{\rm Vec}^{2n}_{\rr{Q}}({\n{T}}^2,\iota)\;\simeq\;{\rm Vec}^{2}_{\rr{Q}}({\n{T}}^2,\iota)\;\simeq\;H^2_{\Z_2}(\n{S}^1|(\n{S}^1)^\iota,\Z(1))\;\simeq\;\Z_2\;.
$$
\begin{lemma}\label{lemma_strong}
The homomorphism
$$
\kappa\;:\;KR^{-3} (\n{T}^2,\iota)\;\longrightarrow\;H^2_{\Z_2}\big(\n{T}^2|(\n{T}^2)^\iota,\Z(1)\big)
$$
described in Theorem \ref{thm:homomorphism_from_KR} is surjective.
\end{lemma}
\proof
The proof follows exactly as in Proposition \ref{prop_strong} by constructing a class $[\n{T}^2\times\C^4, H'_1,H'_0]$ such that 
$\kappa([\n{T}^2\times\C^4, H'_1,H'_0])=1$. In this case the Hamiltonians $H'_1$ and $H'_0$ can be obtained form the 
the Hamiltonians $H_1$ and $H_0$ introduced  in the proof of Proposition \ref{prop_strong} via the pull-back under the equivariant map $\pi_0:(\n{T}^2,\iota)\to (\n{S}^2,\iota)$
described in \cite[eq. (4.14)]{denittis-gomi-14-gen}. More precisely one has that $H'_i(k):=H_i(\pi_0(k))$ for $i=0,1$.
This concludes the proof.
\qed

\medskip

Let $H:{\n{T}^2}\to \rr{Herm}(\C^{4n})^\times$ be a band insulator of class DIII over $(\n{T}^2,\iota)$ represented in a standard way. As in the Dirac case we can define the \emph{strong invariant}
\begin{equation}\label{eq:def_SI_T2}
\nu^{\rm s}(H)\;:=\;\kappa\big([\n{T}^2\times\C^{4n}, H,H_0']\big)
\end{equation}
where 
 $H_0$ is the trivial insulator constructed with the constant sewing matrix $q_0:\n{T}^2\to\n{U}(2n)$ described 
in item (2) of Theorem \ref{thm:invariant}. 
Since this invariant coincides with the FKMM invariant $\kappa(\bb{E}_{q_H})$ of the \virg{Quaternionic} vector bundle associated to $H$ with the sewing matrix $q_H$, one 
can use the  Fu–Kane–Mele index proved in \cite[Theorem 4.2]{denittis-gomi-14-gen} to compute $\nu$.
\begin{proposition}[Fu–Kane–Mele formula for the strong invariant]
The strong invariant $\nu^{\rm s}$ defined by 
\eqref{eq:def_SI_T2} can be computed by the formula 
\begin{equation}\label{eq:strong_tor}
\nu^{\rm s}(H)\;=\;\prod_{k\in(\n{T}^2)^\iota}\frac{{\rm Pf}\big[q_H(k)\big]}{{\rm det}\big[q_H(k)\big]^{\frac{1}{2}}}\;\in\;\{\pm 1\}
\end{equation}
where the branch $k\mapsto {\rm det}[q_H(k)]^{\frac{1}{2}}$ must be chosen continuously and that $\nu^{\rm s}$ must be understood as a map with values into the  multiplicative group $\Z_2\simeq\{\pm1\}$.
\end{proposition}
\proof
According to  \cite[Theorem 4.2]{denittis-gomi-14-gen} one has that
\begin{equation}\label{eq:strong_tor_BIS}
\nu^{\rm s}(H)\;=\;\prod_{k\in(\n{T}^2)^\iota}\frac{{\rm Pf}\big[q_H(k)\big]}{p_H(k)}
\end{equation}
where $p_H:\n{T}^2\to\n{U}(1)$ meets $\det [q_H(k)] = p_H(\iota(k)) p_H(k)$ for all $k \in \n{T}^2$. To obtain \eqref{eq:strong_tor} form \eqref{eq:strong_tor_BIS} we need to show that  one can replace $p_H$ with a map $p_H':\n{T}^2\to\n{U}(1)$
such that $\det [q_H(k)] = p_H'(k)^2$ for all $k \in \n{T}^2$.
The following fact follows by observing that every continuous map 
$s:\n{T}^2\to\n{U}(1)$ such that $s(k)^2=\det [q_H(k)]$ automatically satisfies $s(k)=s(\iota(k))$  for all $k \in \n{T}^2$.
The key to the proof is that every continuous map $f:\n{T}^2\to\n{U}(1)$ can be uniquely expressed as
$$
f(k_1, k_2)\; =\; \expo 
{\ii \left(
\alpha + \widetilde{f}(k_1, k_2) + n_1k_1 + n_2k_2
\right)}\;,\qquad (k_1,k_2)\in\n{T}^2
$$
where $n_1, n_2 \in \Z$, $\alpha \in \R/2\pi\Z$  and $\widetilde{f} : \n{T}^2 \to \R$ is a continuous function such that $\widetilde{f}(0, 0) = 0$. 
The integers $n_1$ and $n_2$ provide the winding numbers of $f$ around the two standard homology cycles of $\n{T}^2$ and $\alpha$ is determined by  $f(0,0)=\expo{ \ii\alpha}$. By applying the same representation to $\det [q_H]$ one gets
$$
\det [q_H(k_1, k_2)]\; =\; \expo 
{\ii \left(
\alpha_H + \widetilde{f}_H(k_1, k_2) + n_1k_1 + n_2k_2
\right)}\;,\qquad (k_1,k_2)\in\n{T}^2
$$
By the invariance $\det[ q_H(k_1, k_2)] = \det[ q_H(-k_1, -k_2)]$, 
and the continuity of $\widetilde{f}_H$ 
it follows that $n_1 = n_2 = 0$ and $\widetilde{f}_H(k_1, k_2) = \widetilde{f}_H(-k_1, -k_2)$. As a result, every continuous function whose square coincides with $\det[ q_H]$
must be of the form
$$
s(k_1, k_2)\; =\: \pm 
\expo{
\ii \left(
{\alpha}_s + \frac{1}{2}\tilde{f}(k_1, k_2)
\right)},
$$
with $\alpha_s = \frac{\alpha_H}{2}\in \R/2\pi\Z$. 
In particular one gets that $s(-k_1, -k_2) = s(k_1, k_2)$ as anticipated.
\qed

\medskip

It is  worth noting equation \eqref{eq:strong_tor} agrees with formula  \cite[eq. (3.69)]{chiu-teo-schnyder-ryu-16}.

\medskip

From Lemma \ref{lemma_strong} and \eqref{eqKR3-T2} one infers that 
$$
{\rm Ker}(\kappa)\;\simeq\;KR^{-3} (\n{T}^2,\iota)/\Z_2\;=\;\Z_2\oplus\Z_2\;.
$$
Theorem \ref{thm:homomorphism_from_KR} establishes the existence the of homomorphism
\begin{equation}\label{eq:nu_T2}
\nu\; :\; {\rm Ker}(\kappa) \; \longrightarrow\; H^1(\n{T}^2,\Z)/H^1_{\Z_2}\big(\n{T}^2|(\n{T}^2)^\iota,\Z(1)\big)\;
\end{equation}
which is of the form $\nu:\Z_2\oplus\Z_2\to\Z_2\oplus\Z_2$ in view of \eqref{eq:comput_cohom_T2}. This homomorphism turns out to be surjective, and therefore bijective.
\begin{lemma}\label{lemma_strong_alpha}
The homomorphism \eqref{eq:nu_T2} 
described in Theorem \ref{thm:homomorphism_from_KR} is an isomorphism.
\end{lemma}

\medskip

Lemma \ref{lemma_strong_alpha} follows as a   special case of Proposition \ref{prop_strong-weak} below.
For that, let us introduce the equivariant inclusions  
$$
\phi_j\;:\; (\n{S}^1,\iota)\;\lhook\joinrel\longrightarrow\; (\n{T}^2,\iota)\;,\qquad j=1,2
$$
defined by $\phi_1(k)=(k,1)$ and $\phi_2(k)=(1,k)$.
To every class $[\n{T}^2\times\C^{2n}, H_0, H_1]\in KR^{-3} (\n{T}^2,\iota)$ we can associate via pull-back two classes
$$
[\n{T}^2\times\C^{2n}, H_0, H_1]_j\;:=\;[\n{S}^1\times\C^{2n}, \phi_j^*H_0, \phi_j^*H_1]\;\in\; KR^{-3} (\n{S}^1,\iota)\;,\qquad j=1,2
$$
where $\phi_j^*H_i(k):=H_i(\phi_j(k))$ for all $k\in \n{S}^1$, $j=1,2$ and $i=0,1$. In view of Proposition \ref{prop:iso-one} one can introduce the homomorphism 
\begin{equation}
\nu^{\rm w}\;:\; KR^{-3}(\n{T}^2,\iota)\; \longrightarrow\;\Big(H^1(\n{S}^1,\Z)/H^1_{\Z_2}\big(\n{S}^1|(\n{S}^1)^\iota,\Z(1)\big)\Big)^{\oplus 2}
\end{equation}
given by
$$
\nu^{\rm w}\;:\;[\n{T}^2\times\C^{2n}, H_0, H_1]\;\mapsto\;\big(\nu([\n{T}^2\times\C^{2n}, H_0, H_1]_1),\nu([\n{T}^2\times\C^{2n}, H_0, H_1]_2)\big)\;.
$$
Notice that $\nu^{\rm w}$ is well-defined on the whole group $KR^{-3}(\n{T}^2,\iota)$ and not only on ${\rm Ker}(\kappa)$.
We are now in position to    provide  a complete classification of 
two-dimensional band insulators of class DIII  in the Bloch case.
\begin{proposition}[Two-dimensional weak-strong invariant]\label{prop_strong-weak}
The homomorphism $\nu^{\rm w}$ extends the homomorphism $\nu$. Moreover the composition $\underline{\nu}:=(\nu^{\rm w},\kappa)$ of the homomorphisms $\nu^{\rm w}$ and $\kappa$ 
provides the isomorphism
$$
\underline{\nu}\;:\;KR^{-3} (\n{T}^2,\iota)\;\longrightarrow\;H^1(\n{T}^2,\Z)/H^1_{\Z_2}\big(\n{T}^2|(\n{T}^2)^\iota,\Z(1)\big)\;\oplus\;H^2_{\Z_2}\big(\n{T}^2|(\n{T}^2)^\iota,\Z(1)\big)\;\;.
$$
\end{proposition}
\proof
The proof is based on the direct calculations. For that let us  construct four sewing matrices $q_{w_1}, q_{w_2}, q_s, q_0$ on  $\n{T}^2$ with values on $\n{U}(2)$ as described below. Let $\pi_1$ and $\pi_2$ be the  
two canonical projections from $\n{T}^2\simeq \n{S}^1\times \n{S}^1$ to $\n{S}^1$.  Then, $q_{w_j}:=q_-\circ \pi_j$, for $j=1,2$, where $q_-:\n{S}^1\to \n{U}(2)$ is the sewing matrix 
introduced in the proof of  Theorem \ref{the:index}.
 The sewing matrix $q_s:=q_1\circ \pi_0$
  is the composition of the equivariant map $\pi_0 : \n{T}^2 \to \n{S}^2$ in the proof of Lemma \ref{lemma_strong} 
  and the sewing matrix $q_1 :\n{S}^2 \to \n{U}(2)$ such that $\kappa(\bb{E}_{q_1}) = 1$ in Proposition \ref{prop_strong}. 
Finally,  $q_0$ is the constant map as in Theorem \ref{thm:invariant} (2). These sewing matrices produce the
 following FKMM invariants:
\begin{align*}
\kappa(\bb{E}_{q_{w_1}}) &\;=\;+ 1\;, &
\kappa(\bb{E}_{q_{w_2}}) &\;=\;+1\;, &
\kappa(\bb{E}_{q_0}) &\;=\;+ 1\;,\; &
\kappa(\bb{E}_{q_s}) &\;=\; -1\;. 
\end{align*}
In addition, one has that 
\begin{align*}
\nu_{q_{w_1} \circ \phi_1} &\;=\; -1\;, &
\nu_{q_{w_2} \circ \phi_1} &\;=\;+ 1\;, &
\nu_{q_0 \circ \phi_1} &\;=\; +1\;, &
\nu_{q_s \circ \phi_1} &\;=\;+ 1\;, \\
\nu_{q_{w_1} \circ \phi_2} &\;=\;+ 1\;, &
\nu_{q_{w_2} \circ \phi_2} &\;=\; -1\;, &
\nu_{q_0 \circ \phi_2} &\;=\; +1\;, &
\nu_{q_s \circ \phi_2} &\;=\;+ 1\;.
\end{align*}
In the above equations the computations for the cases $q_{w_j} \circ \phi_i=q_-\circ(\pi_j\circ\phi_i)$ follows just by observing that  $\pi_j\circ\phi_i$ is the identity map on $\n{S}^1$ when $i=j$, and the restriction to a fixed point if  $i\neq j$.
The cases  $q_0 \circ \phi_i$ with $i=1,2$, are also trivial
in view of the fact that $q_0$ is the constant map.
For the last two cases $q_s \circ \phi_i$, with $i = 1, 2$, we notice that the equivariant map $\pi_0 : \n{T}^2 \to \n{S}^2$ carries the three fixed points $(0, 0), (0, \pi), (\pi, 0)$ to $(1, 0, 0)$ and the fixed point $(\pi, \pi)$ to $(-1, 0, 0)$. We also notice that $q_1$ takes values in $S\n{U}(2) \subset \n{U}(2)$. Since $q_s \circ \phi_i : \n{S}^1 \to S\n{U}(2)$ carries the fixed points of $\n{S}^1$ to the same point of $S\n{U}(2)$, Theorem \ref{thm:invariant_by_Pfaffian} leads to the results above. 
Now, let us denote with $H_{w_1}, H_{w_2}, H_{q_s}$ and $H_0$  the band insulators of class DIII associated to $q_{w_1}, q_{w_2}, q_s$ and $q_0$, respectively. One gets that
$$
\begin{aligned}
\underline{\nu}\big([\n{T}^2 \times \C^2, H_{w_1}, H_0]\big)
\;&=\; (-1, +1, +1)\;, \\
\underline{\nu}\big([\n{T}^2 \times \C^2, H_{w_2}, H_0]\big)
\;&=\; (+1, -1, +1)\;, \\
\underline{\nu}\big([\n{T}^2 \times \C^2, H_{s}, H_0]\big)
\;&=\; (+1, +1, -1),
\end{aligned}
$$
so that $\underline{\nu}$ is a surjection onto $\Z_2 \oplus \Z_2 \oplus \Z_2$. Since $KR^{-3}(T^2) \simeq \Z_2 \oplus \Z_2 \oplus \Z_2$, the homomorphism $\underline{\nu}$ must be an isomorphism. As a result, one obtains that  $\mathrm{Ker}(\kappa)$ is generated by $[\n{T}^2 \times \C^2, H_{w_i}, H_0]$, with $i = 1, 2$. This allows us to see that $\nu^w$ extends $\nu$.\qed

\medskip

The novelty of Proposition \ref{prop_strong-weak} is that, in combination with Theorem \ref{thm:invariant_by_Pfaffian}, it provides a way to compute also the weak invariants
$$
\nu^{\rm w}(H)\;:=\;(\nu_{q_{\phi_1^*H}}, \nu_{q_{\phi_2^*H}})
$$ 
of the two-dimensional band insulator $H$ of class DIII over $(\n{T}^2,\iota)$.

\appendix

\section{Equivariant cohomology and homotopy}\label{sec:eq:cohom_homot}
Let $(X,\tau)$ be an involutive space 
such that $X$ has the structure of a (path connected) $\Z_2$-CW complex. A map $f:X\to\n{U}(1)$  is called \emph{invariant} if $f(x)=f(\tau(x))$ for every $x\in X$. Similarly, it is called \emph{equivariant} if $f(x)=\overline{f(\tau(x))}$ for every $x\in X$. Let $[X, \n{U}(1)]$ be the set of classes of homotopy equivalent functions, $[X, \n{U}(1)]_{\rm inv}$ the set of classes of homotopy equivalent invariant functions and $[X, \n{U}(1)]_{\Z_2}$ the set of classes of homotopy equivalent equivariant functions. These three sets admit a representation in terms of cohomology groups. In fact, it holds true that 
\begin{equation}\label{eq:iso_hom_cohom}
\begin{aligned}
&{[X, \n{U}(1)]}\;&\simeq\;&\;H^1(X,\Z)&\\
&[X, \n{U}(1)]_{\rm inv}\;&\simeq\;&\;H^1_{\Z_2}(X,\Z)&\\
&[X, \n{U}(1)]_{\Z_2}\;&\simeq\;&\;H^1_{\Z_2}(X,\Z(1))&
\end{aligned}
\end{equation}
where $H^\bullet(X,\Z)$ is the usual cohomology ring of $X$ with coefficients in $\Z$, and $H^\bullet_{\Z_2}(X,\Z)$ and $H^\bullet_{\Z_2}(X,\Z(1))$ are the equivariant cohomology groups of $X$ with system of coefficients in $\Z$ and $\Z(1)$, respectively. For a review of the equivariant cohomology we refer to \cite[Section 5]{denittis-gomi-14} and references therein. 
The first isomorphism in \eqref{eq:iso_hom_cohom} is a classical result in algebraic topology 
and follows by observing that $ \n{U}(1)=K(\Z,1)$ is a Eilenberg-Maclane space and in turn $[X, K(\Z,1)]$ is a  representing space for the cohomology $H^1(X,\Z)$ in view of \cite[Theorem 4.57]{hatcher-02}. The proof of the second isomorphism is provided in 
\cite[Lemma A.3]{gomi-thiang-21}, while the third isomorphism is proved in \cite[Proposition A.2]{gomi-15}.

\section{Strong equivalence of \virg{Quaternionic} 
structures}\label{sec:Z2_inv-equiv}

Let ${\rm Vec}^n_{\rr{Q}}({X},\tau)$ be the set of equivalence classes of \virg{Quaternionic} vector bundles of rank $n$ over the involutive space $({X},\tau)$. The subset 
${\rm Vec}^n_{\rr{Q},0}({X},\tau)\subseteq {\rm Vec}^n_{\rr{Q}}({X},\tau)$ consists of the equivalence classes of 
\virg{Quaternionic} vector bundles whose underlying (complex) bundle structure is trivial. 
Let $\bb{E}:=(X\times \C^n, \Theta)$ and $\bb{E}':=(X\times \C^n, \Theta')$ be two \virg{Quaternionic} vector bundles defined on the trivial product bundle  $X\times \C^n$.  The two \virg{Quaternionic} structures 
$\Theta$ and $\Theta'$ are defined by two 
sewing matrices $q$ and $q'$ according to \eqref{eq:sew_01_02}.
The isomorphism condition $\bb{E}\simeq \bb{E}'$
 is equivalent to the existence of a continuous \emph{intertwining} map $\phi:X\to\n{U}(n)$ such that 
\begin{equation}\label{eq:inter_01}
q'(x)\;=\;\phi\big(\tau(x)\big)^\mathtt{t}\; q(x)\; \phi(x)\;,\qquad \forall\; x\in X\;. \qquad (\text{intertwining})
\end{equation}
In this case $\bb{E}$ and $\bb{E}'$ are representatives of the same equivalence class 
 $[\bb{E}]\in {\rm Vec}^n_{\rr{Q},0}({X},\tau)$.
 \begin{definition}[Strong isomorphism]\label{def:Str_iso}
 We will say that the isomorphism between $\bb{E}$ and $\bb{E}'$ is \emph{strong} if the {intertwining} map \eqref{eq:inter_01} meets the extra condition
\begin{equation}\label{eq:inter_02}
\phi(x)\;=\;\phi\big(\tau(x)\big)\;,\qquad \forall\; x\in X\;. \qquad (\tau\text{-invariance})
\end{equation}
In this case we will write $\bb{E}\stackrel{\mathtt{s}}{\simeq} \bb{E}'$.
 \end{definition}

\medskip 

The notion of strong isomorphism introduced with Definition \ref{def:Str_iso} provides a finer classification of ${\rm Vec}^n_{\rr{Q},0}({X},\tau)$. The following example should clarify the difference between the two notions of isomorphism.
\begin{example}
It is known that in the case of the involutive space $({\n{S}}^1,\iota)$ one has that ${\rm Vec}^2_{\rr{Q},0}({\n{S}}^1,\iota)={\rm Vec}^2_{\rr{Q}}({\n{S}}^1,\iota)=0$
\cite[Proposition 2.6]{denittis-gomi-14-gen}
meaning that every rank two \virg{Quaternionic} vector bundle over 
$({\n{S}}^1,\iota)$ is isomorphic to the  trivial \virg{Quaternionic} vector bundle $\bb{E}_0:=({\n{S}}^1\times \C^2, \Theta_{0})$ with \emph{trivial} structure $\Theta_{0}$ induced by the (constant) sewing matrix
$$
q_0(k)\;\equiv\;q_0\;:=\;\left(\begin{array}{cc}0 & -1 \\1 & 0\end{array}\right)\;,\qquad \forall\; k\in{\n{S}}^1\;.
$$
Consider the \virg{Quaternionic} vector bundle $\bb{E}_0':=({\n{S}}^1\times \C^2, \Theta_{0}')$ with \virg{Quaternionic} structure $\Theta_{0}'$ induced by the (again constant) sewing matrix
$$
q'_0(k)\;\equiv\;-q_0\;,\qquad \forall\; k\in{\n{S}}^1\;.
$$
The (constant) intertwining map
$$
\phi_0(k)\;\equiv\;\phi_0\;:=\;\left(\begin{array}{cc}1 & 0 \\0 & -1\end{array}\right)\;,\qquad \forall\; k\in{\n{S}}^1\;.
$$
provides the isomorphism between $\bb{E}_0$ and $\bb{E}_0'$. Moreover, since the map $\phi_0$ is constant, hence $\tau$-invariant, one has that
 the \virg{Quaternionic} vector bundles are \emph{strongly} equivalent, \ie $\bb{E}_0\stackrel{\mathtt{s}}{\simeq} \bb{E}'_0$.
 On the other hand, one can consider the \virg{Quaternionic} vector bundle $\bb{E}_1:=({\n{S}}^1\times \C^2, \Theta_1)$ with \virg{Quaternionic} structure $\Theta_{1}$ induced by the (non constant) sewing matrix
$$
q_1(k)\;:=\;\left(\begin{array}{cc}\sin(k) & -\cos(k) \\ \cos(k) & \sin(k)\end{array}\right)\;,\qquad \forall\; k\in{\n{S}}^1\;.
$$
In this case the isomorphism $\bb{E}_1\simeq \bb{E}_0$ is induced by 
 the intertwining map
$$
\phi_1(k)\;:=\;\expo{\ii\frac{k}{2}}\left(\begin{array}{cc}\sin\left(\frac{k}{2}\right) & - \cos\left(\frac{k}{2}\right) \\\cos\left(\frac{k}{2}\right) & \sin\left(\frac{k}{2}\right) \end{array}\right)\;,\qquad \forall\; k\in{\n{S}}^1\;.
$$
Since $\phi_1(-k)\neq \phi_1(k)$ it follows that the  \virg{Quaternionic} structures $\Theta_{0}$ and $\Theta_{1}$ are only equivalent but \emph{not} strongly equivalent, \ie $\bb{E}_1\stackrel{\mathtt{s}}{\not\simeq} \bb{E}_0$.
\hfill $\blacktriangleleft$
\end{example}

\medskip

In the next result we will use the definitions of unitary equivalence introduced in Definition \ref{def:int_01}. Let us observe that if $V$ is a unitary equivalence between 
two band insulators of class DIII then the equation 
$$
\chi\;V(x)\;=\;V(x)\;\chi
$$
holds true, independently of the nature of $V$.

\begin{proposition}\label{prop:iso_DIII-QV}
Unitarily equivalent band insulators of class DIII define associated \virg{Quaternionic} vector bundles which  are isomorphic. When  the unitary equivalence is strong then the  associated \virg{Quaternionic} vector bundles are strongly isomorphic.
\end{proposition}
\proof
Let $H$ and $H'$ be two unitarily equivalent band insulators of class DIII over  $(X,\tau)$. The functional calculus provides
$Q_{H'}(x)=V(x)Q_{H}(x)V(x)^*$ where the operators $Q_{H'}$ and $Q_{H}$ are defined according to \eqref{eq:sew_01_01-1}. According to Lemma \ref{lemma:diag_form}, in the standard representation described in Preposition \ref{prop:st_rep} the unitary  $V(x)$ takes the form
\begin{equation}
V(x)\;=\;
\left(\begin{array}{cc}\phi_V(x) & 0\\
0 & \overline{\phi_V\big(\tau(x)\big)}\end{array}\right)\;,
\end{equation}
for a suitable map $\phi_V:X\to\n{U}(n)$. Then, the unitary equivalence between $Q_{H}$ and $Q_{H'}$ translates into
$$
q_{H'}(x)\;=\;\overline{\phi_V\big(\tau(x)\big)}\;q_H(x)\;\phi_V(x)^*\;=\;\big[\phi_V\big(\tau(x)\big)^*\big]^\mathtt{t}\;q_H(x)\;\phi_V(x)^*
$$
where the function $q_H$ and $q_{H'}$ are defined according to the off-diagonal decomposition \eqref{eq:sew_01_01}. The latter equation shows that the \virg{Quaternionic} vector bundles associated to $H$ and $H'$ are isomorphic.
In the case of a strong unitary equivalence between $H$ and $H'$, the condition $V(\tau(x))=V(x)$ translates into $\phi_V(\tau(x))=\phi_V(x)$ and in turn in the strong isomorphism between the associated \virg{Quaternionic} vector bundles.
\qed

%bibliography

% 
\end{document}